\documentclass[journal]{IEEEtran}  

\newcommand{\A}{\mathcal{A}}

\usepackage{graphicx}      
\renewcommand{\H}{\mathcal{H}}

\usepackage[table,xcdraw]{xcolor}
\usepackage{amsmath}

\usepackage{amsthm}
\usepackage{tasks}
\usepackage{amssymb}
\usepackage{bbold}
\usepackage[utf8]{inputenc}
\usepackage{amssymb}
\usepackage{dsfont}
\usepackage{adjustbox}
\usepackage{mathtools, cuted}
\usepackage{graphics}
\usepackage{support-caption}
\usepackage{subcaption}
\usepackage{enumitem} 
\renewcommand{\qedsymbol}{$\blacksquare$}

\newtheorem{theorem}{\emph{Theorem}}
\newtheorem{proposition}{Proposition}
\newtheorem{lemma}{Lemma}
\newtheorem{remark}{Remark}

\newtheorem{definition}{Definition}
\newtheorem{example}{Example}
\newtheorem{assumption}{Assumption}
\newtheorem{problem}{Problem}

\title{\LARGE \bf
\textcolor{black}{Zero-Determinant strategies in repeated multiplayer social dilemmas with discounted payoffs}}
\author{Alain Govaert, Ming Cao
\thanks{The work was supported in part by the European Research Council (ERC-CoG-771687) and the Netherlands Organization for Scientific Research (NWO-vidi-14134). }
\thanks{A. Govaert and M. Cao are with
        ENTEG, Faculty of Science and Engineering,
        University of Groningen, The Netherlands, \{a.govaert, m.cao\}@rug.nl.
    }
}

\begin{document}
\maketitle
 \thispagestyle{plain}
\pagenumbering{arabic}
 \pagestyle{plain}

\begin{abstract}
In two-player repeated games, Zero-Determinant (ZD) strategies enable a player to
unilaterally enforce a linear payoff relation between her own and
her opponent’s payoff irrespective of the opponent’s strategy. This
manipulative nature of the ZD strategies attracted significant
attention from researchers due to its close connection to controlling
distributively the outcome of evolutionary games in large
populations. In this paper, necessary and sufficient conditions
are derived for a payoff relation to be enforceable in multiplayer social dilemmas with a finite
expected number of rounds that is determined by a fixed and
common discount factor. Thresholds exist for such a
discount factor above which desired
payoff relations can be enforced. Our results show that depending
on the group size and the ZD-strategist’s initial probability to
cooperate  there exist extortionate, generous and equalizer
ZD-strategies. The threshold discount factors rely on
the desired payoff relation and
the variation in the single-round payoffs.
To show the utility of our results, we apply them
to multiplayer social dilemmas, and
show how discounting affects ZD Nash equilibria.

\end{abstract}



\section{INTRODUCTION}
\IEEEPARstart{T}{he} functionalities of many complex social systems rely on  their composing individuals' willingness to set aside their personal interest for the benefit of the greater good \cite{nowak2011supercooperators}.  
One mechanism for the evolution of cooperation is known as \textit{direct reciprocity}: even if in the short run it pays off to be selfish, mutual cooperation can be favoured when the individuals encounter each other repeatedly. 
Direct reciprocity is often studied in the standard model of repeated games and it is only recently, inspired by the discovery of a novel class of strategies, called zero-determinant (ZD) strategies  \cite{press2012iterated}, that repeated games began to be examined from a new angle by investigating the level of \emph{control} that a single player can exert on the average payoff of the opponents.  
In \cite{press2012iterated} Press and Dyson showed that in \emph{infinitely} repeated $2\times 2$ prisoners dilemma games, if a player can remember the actions in the previous round, this player can \emph{unilaterally} impose some linear relation between his/her own payoff and that of the opponent.
It is emphasized that this enforced linear relation cannot be avoided even if the opponent employs some intricate strategy with a large memories.
Such strategies are called \textit{zero-determinant} because they enforce a part of the transition matrix to have a determinant that is equal to zero.
Later, ZD strategies were extended to games with more than two possible actions \cite{stewart2016evolutionary}, continuous action spaces \cite{mcavoy2017autocratic}, alternative moves \cite{mcavoy2016autocratic}, \textcolor{black}{and observation errors \cite{mamiya2019strategies}. These game-theoretical advancements were subsequently applied to a variety of engineering contexts including cybersecurity in smart grid systems \cite{7007772}, sharing of spectrum bands and resources \cite{6985752,7037497}, and power control of small cell networks \cite{7024762}.}

The success of ZD strategies was also examined from an evolutionary perspective in \cite{hilbe2013evolution,stewart2012extortion}. For a given population size, in the limit of weak selection it was shown in \cite{stewart2013extortion} that all ZD strategies that can survive an invasion of any memory-one strategy must be ``generous'', namely enforcing a linear payoff relation that favors \textit{others}.
This surprising fact was tested experimentally in \cite{hilbe2014extortion}. In \cite{hilbe2018partners} the literature on ZD strategies, direct reciprocity and evolution is reviewed.
Most of the literature focuses on two-player games; 
however, in \cite{pan2015zero} the existence of ZD-startegies in infinitely repeated public goods games was shown by extending the arguments in \cite{press2012iterated} to a symmetric public goods game.
Around the same time, characterization 
of the feasible ZD strategies in multiplayer social dilemmas and those strategies that maintain cooperation in such multiplayer games were reported in  \cite{hilbe2014cooperation}. 
Both in \cite{hilbe2014cooperation} and \cite{pan2015zero} it was noted that group size $n$ imposes restrictive conditions on the set of feasible ZD strategies and that alliances between co-players can overcome this restrictive effect of the group size. The evolutionary success of ZD strategies in such multiplayer games was studied in \cite{hilbe2015evolutionary} and the results show that sustaining large scale cooperation requires the formation of alliances. 
 ZD strategies for repeated $2\times 2$ games with discounted payoffs were defined and characterized in \cite{hilbe2015partners}. In this setting, the discount factor may also be interpreted as a continuation probability that determines the finite number of expected rounds. 
The threshold discount factors above which the ZD strategies can exist were derived in \cite{ichinose2018zero}.
\textcolor{black}{In this paper we use the framework of ZD strategies in infinitely repeated multiplayer social dilemmas from \cite{hilbe2014cooperation} and extend it to the case in which future payoffs are discounted with a fixed and common discount factor, or equivalently, to a repeated game with a finite expected number of rounds. }
We then build upon our results in \cite{govaert2019zero}, in which enforceable payoff relations were characterized, by developing new theory that allows us to express threshold discount factors that determine how \textit{fast} a desired linear payoff relation can be enforced in a multiplayer social dilemma game. \textcolor{black}{These results extend the work of  \cite{ichinose2018zero} to a broad class of multiplayer social dilemmas.}
Our general results are applicable to multiplayer and two player games and can be applied to a variety of complex social dilemma settings including the famous prisoner's dilemma, the public goods game, the volunteer's dilemma, the multiplayer snowdrift game and much more. 
\textcolor{black}{The derived threshold discount factors show how the group size and the payoff functions of the social dilemma affect one's possibilities for exerting control given a constraint on the expected number of interactions, and shows how the discount factor affects Nash equilibria of the repeated game.}
These results can thus be used to investigate, both analytically and experimentally, the effect of the group size and the initial condition on the level of control that a single player can exert in a repeated multiplayer social dilemma game with a finite but undetermined number of rounds.
\textcolor{black}{From an evolutionary perspective, our results may also open the door for novel control techniques that seek to achieve or sustain cooperation in large social systems that evolve under evolutionary forces.\cite{riehl2018survey}}

The paper is organized as follows. 
In section \ref{preliminaries}, preliminaries concerning the game model and strategies are provided.
In section \ref{Mean dist}, the mean distribution of the repeated multiplayer game with discounting and the relation to a memory-one strategy is given. 
In section \ref{fzdsec},  ZD strategies for repeated multiplayer games with discounting are defined, and in section \ref{main results} the enforceable payoff relations are characterized. In section \ref{Threshold}, threshold discount factors are given for generous, extortionate and equalizer ZD strategies. 
We apply our results to the multiplayer linear public goods game and the multiplayer snowdrift game in Section \ref{sec: application to pgg}. In Section \ref{sec: proofs of the main results}, we provide the proofs of our main results.
We conclude the paper in Section \ref{fw}.


\section{Preliminaries}\label{preliminaries}
\subsection{Symmetric multiplayer games}
In this paper we consider multiplayer games in which $n\geq 2$ players can repeatedly choose to either cooperate or defect. 
The set of actions for each player is denoted by $\A=\{C,D\}$.
The actions chosen in the group in round $t$ of the repeated game is described by an action profile $\sigma^t\in \boldsymbol{\A} =\{C,D\}^n$.
A player's payoff in a given round depends on the player's own action and the actions of the $n-1$ co-players. 
In a group in which $z$ co-players cooperate, a cooperator receives payoff $a_z$, whereas a defector receives $b_z$.
As in \cite{hilbe2014cooperation,pan2015zero} we assume the game is symmetric, such that the outcome of the game depends only on one's own decision and the number of cooperating co-players, and hence does not depend on which of the co-players have cooperated. 
Accordingly, the payoffs of all possible outcomes for a player can be conveniently summarized in table \ref{multiplayer payoff}.
\begin{table}[ht]
\centering
\begin{tabular}{ccccccc}
\hline
\rowcolor[HTML]{EFEFEF} 
\begin{tabular}[c]{@{}c@{}}Number of cooperators \\ among co-players\end{tabular} & $n-1$                & $n-2$                & $\dots$              & $2$                  & $1$                  & $0$                  \\ \hline
Cooperator's payoff                                                              & $a_{n-1}$            & $a_{n-2}$            & $\dots$              & $a_{2}$              & $a_{1}$              & $a_{0}$              \\
Defector's payoff                                                                & $b_{n-1}$            & $b_{n-2}$            & $\dots$              & $b_{2}$              & $b_{1}$              & $b_{0}$              \\ \hline
\multicolumn{1}{l}{}                                                             & \multicolumn{1}{l}{} & \multicolumn{1}{l}{} & \multicolumn{1}{l}{} & \multicolumn{1}{l}{} & \multicolumn{1}{l}{} & \multicolumn{1}{l}{}
\end{tabular}
\caption{Single-round payoffs of the symmetric multiplayer game.}
\label{multiplayer payoff}
\end{table}

We have the following assumptions on the single-round payoffs of the symmetric multiplayer game.
\begin{assumption}[Social dilemma assumption \cite{hilbe2014cooperation,kerr2004altruism}]\label{assumptions}
The payoffs of the symmetric multiplayer game satisfy the following conditions:
\begin{tasks}
    \task For all $0\leq z<n-1$, it holds that $a_{z+1}\geq a_z$ and $b_{z+1}\geq b_z$: irrespective of one's own action players prefer other group members to cooperate.
    \task For all $0\leq z<n-1$, it holds that $b_{z+1}>a_z$: within a mixed group defectors obtain strictly higher payoffs than cooperators.
    \task $a_{n-1}>b_0$: mutual cooperation is favored over mutual defection.
\end{tasks}
\end{assumption}
Assumption \ref{assumptions} is standard in multiplayer social dilemma games and ensures that there is an immediate benefit to defect against cooperators, while mutual cooperation leads to better, if not the best, collective outcome.

\begin{example}[Public goods game]\label{example: PGG}
As an example of a game that satisfies Assumption \ref{assumptions}, consider a public goods game in which each cooperator contributes an amount $c>0$ to a public good. 
The sum of the contributions get multiplied by an enhancement factor $1<r<n$ and then divided evenly among all group members. 
The payoff of a cooperator is
$a_z=\frac{rc(z+1)}{n}-c$, while the payoff of a defector is $b_z=\frac{rcz}{n}$, for $z=0,1,\dots,n-1$.
\end{example}


\begin{example}[Multiplayer snowdrift game]\label{example: NSD}
\textcolor{black}{Another example is the multiplayer snowdrift game that traditionally describes a situation in which cooperators need to clear out a snowdrift so that everyone can go on their merry way. By clearing out the snowdrift together, cooperators share a cost $c$ required to create a fixed benefit $b$ \cite{liang2015analysis,souza2009evolution,van2012multi,zheng2007cooperative}.
If a player cooperates together with $z$ group members, their one-shot payoff is $a_z=b-\frac{c}{z+1}$.
If there is at least one cooperator ($z>0$) who clears out the snowdrift, then defectors obtain a benefit $b_z=b$. If no one cooperates, the snowdrift will not be cleared and everyone's payoff is $b_0=0$. }
\end{example}

\subsection{Strategies}
In repeated games the players must choose how to update their actions as the game interactions are repeated over rounds. 
A \textit{strategy} of a player determines the conditional probabilities with which actions are chosen by the player. 
To formalize this concept we introduce some additional notation. 
A history of plays up to round $t$ is denoted by $h^t=(\sigma^0,\sigma^1,\dots,\sigma^{t-1})\in\boldsymbol{\A}^t$ such that each $\sigma^k\in \boldsymbol{\A}$ for all $k=0\dots t-1$. 
The union of possible histories is denoted by $\H=\cup_{t=0}^\infty \boldsymbol{\A}^t$, with $\boldsymbol{\A}^0=\emptyset$ being the empty set.
Finally, let $\Delta(\A)$ denote the probability distribution over the action set $\A$.
As is standard in the theory of repeated games, a strategy of player $i$ is then defined by a function $\rho:\H\rightarrow\Delta(\A)$ that maps the history of play to the probability distribution over the action set.
An interesting and important subclass of strategies are those that only take into account the action profile in round $t-1$, (i.e. $\sigma^{t-1}\in h^t$) to determine the conditional probabilities to choose some action in round $t+1$.
Correspondingly these strategies are called \textit{memory-one strategies}.
The theory of Press and Dyson showed that, for determining the best performing strategies in terms payoffs in two-action repeated games, it is sufficient to consider only the space of memory-one strategies \cite{press2012iterated,stewart2016evolutionary}.

\section{Mean distributions and memory-one strategies in repeated multiplayer games with discounting}\label{Mean dist}
\color{black}In this section we zoom in on a particular player that employs a memory-one strategy in the multiplayer game and refer to this player as the \textit{key player}.
In particular, we focus on the relation between the mean distribution of the action profile and the memory-one strategy of the key player. 
Let $p_{\sigma}\in[0,1]$ denote the probability that the key player cooperates in the next round given that the current action profile is $\sigma\in\boldsymbol{\mathcal{A}}$. 
By stacking the probabilities for all possible outcomes into a vector, we obtain the memory-one strategy $\mathbf{p}=(p_{\sigma})_{\sigma\in\mathbf{\boldsymbol{\mathcal{A}}}}$ whose elements determine the conditional probability for the key player to cooperate in the next round.
Accordingly, the memory-one strategy
$\mathbf{p}_{\sigma}^{\mathrm{rep}},$
gives the probability to cooperate when the current action is simply repeated.
To be more precise, let $\sigma=(\sigma_i,\sigma_{-i})$, where $\sigma_i\in\{C,D\}$ and $\sigma_{-i}\in\{C,D\}^{n-1}$. Then, for all $\sigma_{-i}$, the entries of the repeat strategy are given by $\mathbf{p}_{(C,\sigma_{-i})}^{\mathrm{rep}}=1$ and $\mathbf{p}_{(D,\sigma_{-i})}^{\mathrm{rep}}=0$.
To describe the relation between the memory-one strategy and the mean distribution of the action profile we introduce some additional notation.
Let $v_\sigma(t)$ denote the probability that the outcome of round $t$ is $\sigma\in\boldsymbol{\A}$,
and let $v(t)=(v_{\sigma}(t))_{\sigma\in\boldsymbol{\A}}$ be the vector of these outcome probabilities.
As in \cite{hilbe2015partners,ichinose2018zero,mcavoy2016autocratic,mcavoy2017autocratic} we focus on  repeated games with a finite but undetermined number of rounds.
Given the current round, a fixed and common discount factor, or continuation probability, $0<\delta<1$ determines the probability that a next round takes place.
By taking the limit of the geometric sum of $\delta$, the expected number of rounds is $\frac{1}{1-\delta}$.
As in \cite{hilbe2015partners}, the mean distribution of $v(t)$ is:
    \begin{equation}\label{limitv}
    \mathbf{v}=(1-\delta)\sum^\infty_{t=0}\delta^tv(t).     
    \end{equation}
 
 In this paper we are interested in the expected and discounted payoffs of the players in the repeated game. 
Let $g^i_{\sigma}$ denote the single-round payoff that player $i$ receives in the action profile $\sigma$. 
 The vector $g^i=(g^i_{\sigma})_{\sigma\in\mathbf{\boldsymbol{\mathcal{A}}}}$ thus contains all possible payoffs of player $i$ in a given round. 
 The \textit{expected} single-round payoff of player $i$ in round $t$ is then given by
     $\pi^i(t)= g^i\cdot v(t). $
    The \textit{average discounted payoff} of player $i$ in the repeated game is then \cite{mailath2006repeated} \color{black} 
    \begin{equation*}\label{FRFP}
    \begin{aligned}
         \pi^i=(1-\delta)\sum_{t=0}^\infty \delta^t\pi^i(t)
         =(1-\delta)\sum_{t=0}^\infty \delta^tg^i\cdot v(t)=g^i\cdot\mathbf{v}.
    \end{aligned}
    \end{equation*}
    The following lemma relates the limit distribution $\mathbf{v}$ to the memory-one strategy $\mathbf{p}$ of the key player. 
    The presented lemma is a straightforward multiplayer extension of the 2-player case that is given in \cite{hilbe2015partners} and relies on the fundamental results from \cite{akin2016iterated}. 
    \begin{lemma}[A fundamental relation] \label{limit dist}
    Suppose the key player applies memory-one strategy $\mathbf{p}$ and the strategies of the other players are arbitrary, but fixed.  
    Then, it holds that
    \begin{equation}\label{eq: Akin Lemma}
       (\delta\mathbf{p}-\mathbf{p}^{\mathrm{rep}})\cdot\mathbf{v}= -(1-\delta)p_0,
    \end{equation}
    where $p_0$ is the key player's initial probability to cooperate.
    \end{lemma}
    
    \begin{IEEEproof}
    The probability that $i$ cooperated in round $t$ is
    $q_C(t)= \mathbf{p}^{\mathrm{rep}}\cdot v(t).$
    And the probability that $i$ cooperates in round $t+1$ is
    $q_C(t+1)= \mathbf{p}\cdot v(t).$
    Now define, 
    \begin{equation}\label{ue}
         u(t):=\delta q_C(t+1)-q_C(t)=(\delta\mathbf{p}-\mathbf{p}^{\mathrm{rep}})\cdot v(t).
    \end{equation}
    Multiplying equation \eqref{ue} by $(1-\delta)\delta^t$ and summing up over $t=0,\dots,\tau$ we obtain the \textit{telescoping} sum  
   $(1-\delta)\sum^\tau_{t=0}\delta^tu(t)=(1-\delta)\delta^{\tau+1}q_c(\tau+1)-(1-\delta)q_c(0)$.
   Because $0<\delta<1$, it follows that
   $\underset{\tau\rightarrow\infty}{\text{lim}} (1-\delta)\sum^\tau_{t=0}\delta^tu(t)= -(1-\delta)p_0.$
The result follows by substituting the definition of $u(t)$ and $\mathbf{v}$. That is,
   $\underset{\tau\rightarrow\infty}{\text{lim}}(1-\delta)\sum^\tau_{t=0}\delta^t(\delta\mathbf{p}-\mathbf{p}^{\mathrm{rep}})\cdot v(t)=(\delta\mathbf{p}-\mathbf{p}^{\mathrm{rep}})\cdot\mathbf{v} =-(1-\delta)p_0.$
   \end{IEEEproof}
   
   \begin{remark}
   Note that in the limit $\delta\rightarrow 1$, the infinitely repeated game is recovered.
   In this setting, the expected number of rounds is infinite.
   If the limit exists, the payoff is given by
    $ \pi^i=\underset{\tau\rightarrow\infty}{\text{lim}}\frac{1}{\tau+1}\sum_{t=0}^\tau \pi^i(t).$
   By Akins Lemma (see \cite{akin2016iterated,hilbe2014cooperation}), for the repeated game without discounting, irrespective of the initial probability to cooperate, it holds that
   $(\mathbf{p}-\mathbf{p}^{\mathrm{rep}})\cdot \mathbf{v}= 0.$
  Hence, a key difference between the repeated game with and without discounting is that $p_0$ remains important for the relation between the memory-one strategy $\mathbf{p}$ and the mean distribution $\mathbf{v}$ when the game has a finite number of expected interactions.
  In the limit $\delta\rightarrow 1$, the importance of the initial conditions on the relation between $\mathbf{p}$ and $\mathbf{v}$ disappears \cite{hilbe2014cooperation}. 
   \end{remark}

  
  \section{ZD-strategies in multiplayer games with discounting}\label{fzdsec}
  \color{black}We now investigate the effect that the key player's memory-one strategy can have on the average discounted payoffs in the repeated game. We will use $i$ to indicate the key player and $j$ to indicate his/her co-players.
  Let $g^j_\sigma$ denote the single-round payoff of player $j$ in action profile $\sigma\in\boldsymbol{\mathcal{A}}$, and let $g^j=(g^j_\sigma)_{\sigma\in\boldsymbol{\mathcal{A}}}$ be the vector of these payoffs. 
  Based on Lemma \ref{limit dist} we now formally define a ZD strategy for a multiplayer game with discounting. For this we let $\mathds{1}=(1)_{\sigma\in\mathbf{\boldsymbol{\mathcal{A}}}}$.  
  
\color{black}

\begin{definition}[ZD strategy]\label{ZDrewritten}
A memory-one strategy $\mathbf{p}$ with all entries in the closed unit interval is a ZD strategy if there exist constants $s,l,\phi$ and weights $w_j$, for $j\neq i$  such that
\begin{align}\label{eq:nfZD}
      \delta\mathbf{p}=\mathbf{p}^{\mathrm{rep}}+\phi\left[sg^i -\sum_{j\neq i}^nw_jg^j+(1-s)l\mathds{1}\right]-(1-\delta)p_0\mathds{1},
\end{align}
 under the conditions that $\phi\neq 0$ and $\sum_{j\neq i}^nw_j=1$. 
\end{definition}

\begin{remark}
\textcolor{black}{
When $\delta=1$, the ZD strategy in Definition \eqref{ZDrewritten} recovers the ZD strategies studied in \cite{hilbe2014cooperation}. We will elaborate on the effect of the discount factor $\delta$ in Sections \ref{Threshold} and \ref{sec: application to pgg}.}
\end{remark}

 \begin{remark}
 \textcolor{black}{
When all weights are equal, i.e. $w_j=\frac{1}{n-1}$ for all $j\neq i$, the formulation of the ZD strategy for a symmetric multiplayer social dilemma can be simplified using only the number of cooperators in the social dilemma.
To this end, let $g^{-i}_{\sigma_i,z}$ denote the average single-round payoff of the $n-1$ co-players of $i$ when player $i$ selects action $\sigma_i\in\{C,D\}$ and $0\leq z\leq n-1$ co-players cooperate. 
Using the single-round payoffs in Table \ref{multiplayer payoff} this can be written as $g^{-i}_{C,z}=\frac{a_{z}z+(n-1-z)b_{z+1}}{n-1}$, and  $g^{-i}_{D,z}=\frac{a_{z-1}z+(n-1-z)b_{z}}{n-1}$. We obtain $g^{-i}=(g^{-i}_{\sigma_i,z})$ by stacking these payoffs into a vector.
Similarly, let $v_{\sigma_i,z}(t)$ denote the probability that at round $t$, player $i$ chooses action $\sigma_i$ and $0\leq z\leq n-1$ co-players cooperate, and let $v(t)=(v_{\sigma_i,z}(t))\in[0,1]^{2n}$ be the vector of these outcome probabilities. 
The expected payoff of player $i$ at time $t$ is again given by $\pi^i(t)= g^i\cdot v(t)$.  Moreover, the average expected payoff of the co-players at time $t$ can be conveniently written as $\pi^{-i}(t)= g^{-i}\cdot v(t)$.
The mean distribution of $v(t)$ is again obtained using \eqref{limitv}, but now the entries of $\mathbf{v}$ provide the fraction of rounds in the repeated game in which player $i$ chooses $\sigma_i$ and $z$ players cooperate.  Then, as before, $\pi^i=g^{i}\cdot\mathbf{v}$ and  $\pi^{-i}=g^{-i}\cdot\mathbf{v}$ which leads to the ZD strategy
$$\delta\mathbf{p}=\mathbf{p}^{\mathrm{rep}}+\alpha g^i+ g^{-i}+(\gamma -(1-\delta)p_0)\mathds{1}_{2n}.$$}
\end{remark}
Let $w=(w_i)\in\mathbb{R}^{n-1}$ denote the vector of weights that the ZD strategist assigns to her co-players. 
The following proposition shows how the ZD strategy can enforce a linear relation between the key player's  average discounted payoff (from now on simply called payoff) and a weighted average payoff of his/her co-players.
\begin{proposition}[Enforcing a linear payoff relation]
Suppose the key player employs a fixed ZD strategy with parameters $s$, $l$ and weights $w$ as in Definition \ref{ZDrewritten}. Then, irrespective of the fixed strategies of the remaining $n-1$ co-players, payoffs obey the equation
     \begin{equation}\label{eq: payoff relation s l}
        \pi^{-i}=s\pi^i+(1-s)l,
    \end{equation}
    where $\pi^{-i}=\sum_{j\neq i}^nw_j\pi^j$.
\end{proposition}

\begin{IEEEproof}
The proof follows by substituting \eqref{eq:nfZD} into \eqref{eq: Akin Lemma}. 
\end{IEEEproof}

   
  \textcolor{black}{
    For positive slopes $s$ and weights $w_j>0$ for all $j\neq i$, the linear payoff relation in \eqref{eq: payoff relation s l} ensures that the \textit{collective best response} of the co-players also maximizes the benefits of the key player.
    This is particularly interesting in social dilemmas in which the payoff increases with the number of cooperating co-players.
  The strength of the correlation between the payoffs is determined by the slope $s$ of the linear payoff relation. 
    For positive slopes $0<s<1$, the baseline payoff results in a generous ($l=a_{n-1}$) or extortionate ($l=b_0$) payoff relation. 
    The former typically implies a relative performance in which co-players, on average, 
    do better than the ZD strategist ($\pi^{-i}\leq \pi^{i}$), while the latter typically implies the ZD strategist outperforms the average payoff of his/her co-players ($\pi^{-i}\leq \pi^{i}$) \cite{hilbe2014cooperation}.  
    The theoretical ability of generous and extortionate ZD strategies to promote selfless cooperation of co-players was empirically studied in \cite{hilbe2014extortion,wang2016extortion}.
   Two special cases of the linear payoff relation remain of interest. 
   When $s=1$ the average payoff of the co-players is equal to the payoff of the key player. Such ZD strategies are called fair and were proven to exist in infinitely repeated social dilemmas in \cite{hilbe2014cooperation}. 
   In the other extreme $s=0$,  payoffs are not correlated but the key player can set the average payoff of the co-player to the baseline payoff $l$. Table \ref{4 ZD} summarizes the most studied ZD strategies.}

 \begin{table}[ht]
  \caption{The four most studied ZD strategies.}\label{4 ZD}
    \centering
    \begin{tabular}{ccc}
    \rowcolor[HTML]{EFEFEF} 
ZD strategy  & Parameter values    & Enforced payoff relation \\ \hline
Fair         & $s=1$               & $\pi^{-i}=\pi^i$         \\
Generous     & $l=a_{n-1}$, $0<s<1$ & $\pi^{-i}=s\pi^i+(1-s)a_{n-1}$      \\
Extortionate & $l=b_0$, $0<s<1$    & $\pi^{-i}=s\pi^i+(1-s)b_0$      \\
Equalizer    & $s=0$               & $\pi^{-i}=l$             \\ \hline
    \end{tabular}
\end{table}

  \textcolor{black}{ Because the entries of the ZD strategy correspond to conditional probabilities, they are required to belong to the unit interval and not every linear payoff relation with parameters $s,l$ is can be enforced. 
  For repeated games with discounting, the discount factor that determines the expected number of rounds is part of the ZD strategy and therefore influences the set of enforceable payoff relations. 
  We will focus on the role of the discount factor $\delta$ in the remainder of the paper.}
    Consider the following definition that was given in \cite{hilbe2015partners}
    for two-player games.
    \begin{definition}[Enforceable payoff relations]\label{def: enforceable}
Given a discount factor $0<\delta<1$, a payoff relation $(s,l)\in\mathbb{R}^2$ with weights $w$ is enforceable if there are $\phi\in\mathbb{R}$ and $p_0\in[0,1]$, such that each entry in $\delta\mathbf{p}$ according to equation \eqref{eq:nfZD} is in $[0,\delta]$. We indicate the set of enforceable payoff relations by $\mathcal{E}_\delta$.
\end{definition}

An intuitive implication of decreasing the expected number of rounds in the repeated game (by decreasing $\delta$) is that the set of enforceable payoff relations will decrease as well. 
This monotone effect is formalized in the following proposition that extends a result from \cite{hilbe2015partners} to the multiplayer case.
\begin{proposition}[Monotonicity of $\mathcal{E}_\delta$]\label{monotone}
If $\delta'\leq\delta''$, then $\mathcal{E}_{\delta'}\leq\mathcal{E}_{\delta''}.$
\end{proposition}

\begin{IEEEproof}
Albeit with different formulations of $\mathbf{p}$, the proof follows from the same argument used in the the two-player case \cite{hilbe2014cooperation}. It is presented here to make the paper self-contained.
From Definition \ref{def: enforceable}, $(s,l)\in\mathcal{E}_\delta$ if and only if one can find $\phi\in\mathbb{R}$ and $p_0\in[0,1]$ such that $\mathbf{p}\in[0,1]^n$. Let $\mathbb{0}=(0)_{\sigma\in\mathbf{\mathcal{A}}}$, we have
\begin{equation}
\begin{aligned}
     \mathbb{0}\leq\mathbf{p}\leq \mathds{1} \Rightarrow
    \mathbb{0}\leq\delta\mathbf{p}\leq \delta\mathds{1}.
\end{aligned}   
\end{equation}
Then by substituting \eqref{eq:nfZD} into the above inequality we obtain,
 \begin{equation}\label{bounds on infty ZD}
    p_0(1-\delta)\mathds{1}\leq\mathbf{p}^{\infty}\leq\delta\mathds{1}+(1-\delta)p_0\mathds{1}, 
 \end{equation}
with $$\mathbf{p}^{\infty}=\mathbf{p}^{\mathrm{rep}}+\phi\left[sg^i -\sum_{j\neq i}^nw_jg^j+(1-s)l\mathds{1}\right].$$
Now observe that $p_0(1-\delta)\mathds{1}$ on the left-hand side of the inequality \eqref{bounds on infty ZD} is decreasing for increasing $\delta$. Moreover, $\delta\mathds{1}+(1-\delta)p_0\mathds{1}$ on the right-hand side of the inequality is increasing for increasing $\delta$.  The middle part of the inequality, which is exactly the definition of a ZD strategy for the infinitely repeated game in \cite{hilbe2014cooperation}, is independent of $\delta$. It follows that by increasing $\delta$ the range of possible ZD parameters $(s,l,\phi)$ and $p_0$ increases and hence if $\mathbb{0}\leq\mathbf{p}\leq\mathds{1}$ is satisfied for some $\delta'$, then it is also satisfied for some $\delta''\geq\delta'$.
\end{IEEEproof}

We are now ready to state the existence problem studied in this paper.

\begin{problem}[The existence problem]
For the class of multiplayer social dilemmas with payoffs as in Table \ref{multiplayer payoff} that satisfy Assumption \ref{assumptions}, what are the enforceable payoff relations when the expected number of rounds is finite, i.e., $\delta\in(0,1)$?
\end{problem}

\textcolor{black}{Characterizing the set of enforceable payoff relations is important not only because it describes the possibilities for a single player to exert control in the repeated game, but also because it allows us to characterize the equilibrium set for ZD strategies. If all weights are equal and players apply the same ZD strategy, then all players' payoff is $l$. 
The incentive to deviate from the common ZD strategy can then be analysed with respect to the enforced baseline payoff $l$ and the enforced linear payoff relations to obtain Nash equilibrium conditions. 
If the set of enforceable payoff relations includes the minimum and maximum average group payoff per round, then the Nash equilibrium conditions can be extended to arbitrary ``mutant'' strategies \cite{hilbe2014cooperation}. 
Including the discount factor in the characterization of the enforceable payoff relations thus allows to explain how Nash equilibria of the repeated social dilemma can change under the influence of discounting. 
In Section \ref{sec: application to pgg}, we return to this using Example \ref{example: PGG} and \ref{example: NSD}. }

\section{Existence of ZD strategies in symmetric multiplayer social dilemmas with discounting}\label{main results}

  In this section, we present our results on the existence problem. The proofs of these results are found in Section \ref{sec: proofs of the main results}.
  We begin by formulating conditions on the parameters of the ZD strategy that are necessary for the payoff relation to be enforceable in the finitely repeated multiplayer game.
  
\begin{proposition}[Necessary conditions]\label{necessary}
The enforceable payoff relations $(l,s,w)$ in the repeated multiplayer game with $\delta\in(0,1)$ and single-round payoffs as in Table \ref{multiplayer payoff} that satisfy Assumption \ref{assumptions}, require 
$ -\frac{1}{n-1}\leq-\underset{j\neq i}{min}\ w_j < s< 1,$
$\phi>0,$ and
    $b_0\leq l \leq a_{n-1}$,
with at least one strict inequality. 
\end{proposition}


In the following theorem we extend the results from \cite{hilbe2014cooperation} to multiplayer social dilemmas with discounting. To write the statement compactly, we let $a_{-1}=b_{n}=0$. Moreover, let $\hat{w}_z=\underset{w_h\in w}{\text{min}}(\sum_{h=1}^zw_h)$ denote the sum of the $z$ smallest weights and let $\hat{w}_0=0$.

\begin{theorem}[Characterizations of enforceable payoff relations]\label{Enforceable}
 For the repeated multiplayer game with one-shot payoffs as in Table \ref{multiplayer payoff} that satisfy Assumption \ref{assumptions}, the payoff relation $(s,l)\in\mathbb{R}^2$ with weights $w\in\mathbb{R}^{n-1}$ is enforceable for some $\delta\in(0,1)$ if and only if $-\frac{1}{n-1}<s<1$ and
\begin{equation}\label{l bounds complete}
\begin{aligned}
     \underset{0\leq z\leq n-1}{\text{max}}\left\{ b_z-\frac{\hat{w}_z(b_z-a_{z-1})}{(1-s)} \right\} &\leq l, \\
     \underset{0\leq z\leq n-1}{\text{min}}\left\{ a_{z}+\frac{\hat{w}_{n-z-1}(b_{z+1}-a_{z})}{(1-s)} \right\}  &\geq l,
\end{aligned}
\end{equation}
moreover, at least one inequality in \eqref{l bounds complete} is strict.
\end{theorem}

\begin{remark}
For $n=2$ the full weight is placed on the single opponent i.e., $\hat{w}_j=1$. When the payoff parameters are defined as $b_1=T$, $b_0=P$, $a_1=R$, $a_0=S$, the result in Theorem \ref{Enforceable} recovers the earlier result in \cite{hilbe2015partners}.
\end{remark}


\textcolor{black}{An immediate consequence of Theorem \ref{Enforceable} is that fair strategies with $s=1$, that \textit{always} exist in infinitely repeated social dilemmas without discounting \cite{hilbe2014cooperation}, \textit{never} exist in these games when payoffs are discounted. For example, proportional Tit-for-Tat, that is a fair ZD strategy for the infinitely repeated public goods game \cite{hilbe2014cooperation}, is \textit{not} a ZD strategy in the repeated public goods game with discounting. 
In Section \ref{sec: application to pgg}, we apply our results to two well-known multiplayer social dilemmas to illustrate the crucial role of $\delta$ in the possibilities for enforcing a linear payoff relation.}

Theorem \ref{Enforceable} does not stipulate any conditions on the key player's initial probability to cooperate other than $p_0\in[0,1]$. However, the existence of extortionate and generous strategies does depend on the value of $p_0$. This is formalized in the following Lemma that was also observed in \cite{hilbe2014cooperation,ichinose2018zero}.

\begin{lemma}[Necessary conditions on $p_0$]\label{p0con}
For the existence of extortionate strategies it is necessary that $p_0=0$.
 Moreover, for the existence of generous ZD strategies it is necessary that $p_0=1$.
\end{lemma}

These requirements on the key player's initial probability to cooperate make intuitive sense. 
In a finitely repeated game, if the key player aims to be an \textit{extortioner} that profits from the cooperative actions of others, she cannot start to cooperate because she could be taken advantage off by defectors.  On the other hand, if she aims to be \textit{generous}, she cannot start as a defector because this will punish both cooperating and defecting co-players. 
\textcolor{black}{ 
The requirements on the key player's initial probability to cooperate are also useful in characterizing the effect of the discount factor $\delta$ on the set of enforceable slopes. This will be investigated in the next section.}

 \section{Thresholds on discount factors}\label{Threshold}
In the previous section we have characterized the enforceable payoff relations of ZD strategies in multiplayer social dilemma games with discounted payoffs. 
Our conditions generalize those obtained for two-player games and illustrate how a single player can exert control over the outcome of a multiplayer social dilemma with a finite number of expected rounds. 
The conditions that result from the existence problem do not specify requirements on the discount factor other than $\delta\in(0,1)$. 
However, as we will see the discount factor or ``patience'' of the players in the multiplayer social dilemma heavily influences the possibilities to exert control in the repeated multiplayer social dilemma.
Threshold discount factors, above which a payoff relation can be enforced, provide insight into the minimum number of expected interactions that are required to enforce a desired linear payoff relation. 
In this section we address the following problem, that was studied for two player games in \cite{ichinose2018zero}.
 
 \begin{problem}[The minimum threshold problem]
Suppose the desired payoff relation $(s,l)\in\mathbb{R}^2$ satisfies the conditions in Theorem \ref{Enforceable}. What is the minimum $\delta\in(0,1)$ under which the linear relation $(s,l)$ with weights $w$ can be enforced by the ZD strategist?
 \end{problem}

 We consider the three classes of ZD strategies separately.
 Before giving the main results it is necessary to introduce some additional notation. 
Define $\tilde{w}_z=\underset{w_h\in w}{\text{max}}\sum_{h=1}^z w_h$ to be the maximum sum of weights for some permutation of $\sigma\in\boldsymbol{\A}$ with $z$ cooperating co-players. 
Additionally, for a payoff relation $(s,l)\in\mathbb{R}^2$ and weights $w\in\mathbb{R}^{n-1}$ define 
\begin{align}\label{eq: rho functions}
    \overline{\rho}^C&:=\underset{\scriptscriptstyle0\leq z\leq n-1}{\text{max}} (1-s)(a_z-l)+\tilde{w}_{n-z-1}(b_{z+1}-a_z),\nonumber \\
\underline{\rho}^C&:=\underset{\scriptscriptstyle0\leq z\leq n-1}{\text{min}} (1-s)(a_z-l)+\hat{w}_{n-z-1}(b_{z+1}-a_z), \nonumber\\
 \overline{\rho}^D&:=\underset{\scriptscriptstyle0\leq z\leq n-1}{\text{max}}(1-s)(l-b_z)+\tilde{w}_{z}(b_z-a_{z-1}),\nonumber\\
\underline{\rho}^D&:=\underset{\scriptscriptstyle0\leq z\leq n-1}{\text{min}}(1-s)(l-b_z)+\hat{w}_{z}(b_z-a_{z-1}).
\end{align}
In the following, we will use these extrema to derive threshold discount factors for extortionate, generous and equalizer strategies in symmetric multiplayer social dilemma games. The proofs of our results can be found in Section \ref{sec: proofs of the main results}.

 \subsection{Extortionate ZD strategies}
 We first consider the case in which $l=b_0$ and $0<s<1$, such that the ZD strategy is extortionate. We have the following result.
 
 \begin{theorem}[Extortion thresholds]\label{prop: extortionate thresholds}
Assume $p_0=0$ and the payoff relation $(s,b_0)\in\mathbb{R}^2$ satisfies the conditions in Theorem \ref{Enforceable}, then $\overline{\rho}^C>0$ and  $\overline{\rho}^D+\underline{\rho}^C>0$.
Moreover, the threshold discount factor above which the extortionate payoff relation can be enforced is determined by
 \begin{equation*}
         \delta_\tau=\text{max}\left\{\frac{\overline{\rho}^C-\underline{\rho}^C}{\overline{\rho}^C}, \frac{\overline{\rho}^D}{\overline{\rho}^D+\underline{\rho}^C} \right\}.
 \end{equation*}
 \end{theorem}

\subsection{Generous ZD strategies}
If a player instead aims to be generous, in general, different thresholds will apply. Thus, we now consider the case in which $l=a_{n-1}$ and $0<s<1$ such that the ZD strategy is generous. 
 \begin{theorem}[Generosity thresholds]\label{prop: generous thresholds}
Assume $p_0=1$ and the payoff relation $(s,a_{n-1})\in\mathbb{R}^2$ satisfies the conditions in Theorem \ref{Enforceable}. 
Then $\overline{\rho}^D>0$ and $\overline{\rho}^C+\underline{\rho}^D>0$. Moreover, the threshold discount factor above which the generous payoff relation can be enforced is determined by
 \begin{equation*}
      \delta_\tau=\text{max}\left\{\frac{\overline{\rho}^D-\underline{\rho}^D}{\overline{\rho}^D}, \frac{\overline{\rho}^C}{\overline{\rho}^C+\underline{\rho}^D} \right\}.
 \end{equation*}
 \end{theorem}
 \subsection{Equalizer ZD strategies}
The existence of equalizer strategies with $s=0$ does not impose any requirement on the initial probability to cooperate. In general, one can identify different regions of the unit interval for $p_0$ in which different threshold discount factors exist. 
For instance, the boundary cases can be examined in a similar manner as was done for extortionate and generous strategies and, in general, will lead to different requirements on the discount factor.
In this section, we derive conditions for the discount factor such that the equalizer payoff relation can be enforced for a variable initial probability to cooperate that is within the open unit interval.
\begin{theorem}[Equalizer thresholds]\label{thm: equalizer strategies thresholds}
	Let $s=0$ and assume $l$ satisfies the bounds in Theorem \ref{Enforceable}. The equalizer payoff relation can be enforced for $p_0\in(0,1)$ if and only if the following inequalities hold
	\begin{align}
	\delta&\geq 1-\frac{\underline{\rho}^D}{\underline{\rho}^D+(\overline{\rho}^D-\underline{\rho}^D)p_0}\label{req1},\\
	\delta&\geq 1-\frac{\underline{\rho}^C}{(1-p_0)(\underline{\rho}^C+\overline{\rho}^D)}\label{req2},\\
	\delta&\geq 
	1-\frac{\underline{\rho}^C}{(1-p_0)(\overline{\rho}^C-\underline{\rho}^C)+\underline{\rho}^C}\label{req3},\\
	\delta&\geq 1-\frac{\underline{\rho}^D}{\left(\overline{\rho}^C+\underline{\rho}^D\right)p_0}\label{req4}.
	\end{align}
	\textcolor{black}{In this case, $\delta_\tau$ is determined by the maximum right-hand-side of \eqref{req1}-\eqref{req4}.} These conditions on $\delta$ also hold when $s\neq 0$ and $b_0<l<a_{n-1}$.
\end{theorem}

\begin{remark}
\textcolor{black}{Because the maxima and minima in \eqref{eq: rho functions} depend on the slope $s$, the expressions of the threshold discount factors for a fixed baseline payoff $l$ typically varies over the set of enforceable slopes. This is exemplified in Section \ref{sec: application to pgg}.
However, the expressions of the threshold discount factors also provide insight into why fair payoff relations $\pi^{-i}=\pi^i$ with a slope $s=1$ cannot be enforced in a repeated social dilemma with a finite expected number of rounds. 
By Assumption \ref{assumptions}b and $\hat{w}_0=0$ it follows that both $\underline{\rho}^D$ and $\underline{\rho}^C$ in \eqref{eq: rho functions} are zero when $s=1$. As a result, all expressions for $\delta_\tau$ are equal to one.}
\end{remark}

With Theorems \ref{prop: extortionate thresholds}, \ref{prop: generous thresholds}, and \ref{thm: equalizer strategies thresholds}, we have provided expressions for deriving the minimum discount factor for some desired linear payoff relation. 
Because the expressions depend on the `single-round payoff of the multiplayer game, in general they will differ between social dilemmas. 
In order to determine the thresholds, one needs to find the global extrema of a function over $z$ that, as we will show in the next section, can be efficiently done for a many social dilemma games.
\textcolor{black}{Essentially, the obtained threshold discount factors ensure that a suitable $\phi>0$ exists for which the ZD strategy is well-defined. 
Section \ref{sec: proofs of the main results} contains a detailed derivation of the thresholds that also indicates how to set $\phi$ to fully define the ZD strategy in terms of the game parameters and the desired enforceable payoff relation. }

\section{Applications to multiplayer social dilemmas}\label{sec: application to pgg}

 
\textcolor{black}{In this section the above theory is applied to the linear public goods game of Example \ref{example: PGG} and the multiplayer snowdrift game of Example \ref{example: NSD} to illustrate the role of the discount factor on the set of enforceable slopes $s$ for generous and extortionate strategies and subsequently, the Nash equilibria of the repeated game. Characterizing this effect is important also because the slope 
determines the correlation between the payoffs and thus also the degree to which cooperative actions of opponents can be incentivised \cite{stewart2013extortion} within a finite expected number of rounds.
 The weights are assumed to be equal, that is $w_j=\frac{1}{n-1}$ for all $j\neq i$. 
 In this case, the conditions for existence and the thresholds become relatively easy to obtain. 
All the proofs of this section are found in the Appendices.}
 We first apply Theorem \ref{Enforceable} to the public goods game to characterize the enforceable slopes and baseline payoffs. 
 
 \begin{proposition}[Enforceable slopes in the public goods game]\label{prop: ex. ext. pgg}
Suppose $p_0=0$, $l=0$ and $0<s<1$, so that the ZD strategy is extortionate. For the public goods game with discounting and $r>1$, every slope $s\geq \frac{r-1}{r}$ can be enforced independent of $n$. If $s< \frac{r-1}{r}$, the slope can be enforced if and only if
$$n\leq\frac{r(1-s)}{r(1-s)-1}.$$
Generous strategies with $p_0=1$ and $l=rc-c$ have the same set of enforceable slopes.
\end{proposition}


\textcolor{black}{Extortionate strategies in the public goods game (and the multiplayer snowdrift game) satisfy $l=b_0=0$ and thus the enforced linear payoff relation simply becomes $\pi^{-i}=s\pi^i$. Slopes close to one thus imply $\pi^{-i}$ and $\pi^i$ are approximately equal, while slopes close to zero imply a high level of extortion that allows the strategic player to do better than the average of his/her co-players.
From Proposition \ref{prop: ex. ext. pgg} it follows that in the public goods game the lower bound on enforceable slope is $s\geq 1-\frac{n}{r(n-1)}$.
Both $n$ and $r$ thus determine how much better the strategic player can do than the average of his/her co-players.  
Because full cooperation leads to the highest single-round average group payoff an opposite argument can be made for generous strategies: low values of $s$ ensure the average payoff of the co-players is close to optimal.}

\textcolor{black}{
Just like the set of enforceable slopes, also the threshold discount factors for generous and extortionate strategies are the same in the public goods game and are characterized in the following proposition.}
\begin{proposition}[Thresholds for extortion and generosity]\label{prop: thresholds pgg}
 For the enforceable slopes $s\geq 1-\frac{n}{r(n-1)}$, in the public goods game the threshold discount factor for extortionate and generous strategies is determined as
  \begin{equation}\label{eq: threshold pgg}
     \delta_\tau=\frac{1-(1-s)(r-\frac{r}{n})}{1-(1-s)(1-\frac{r}{n})}.
 \end{equation}
\end{proposition}

 \begin{figure}
 \centering
 \includegraphics[width=0.8\linewidth]{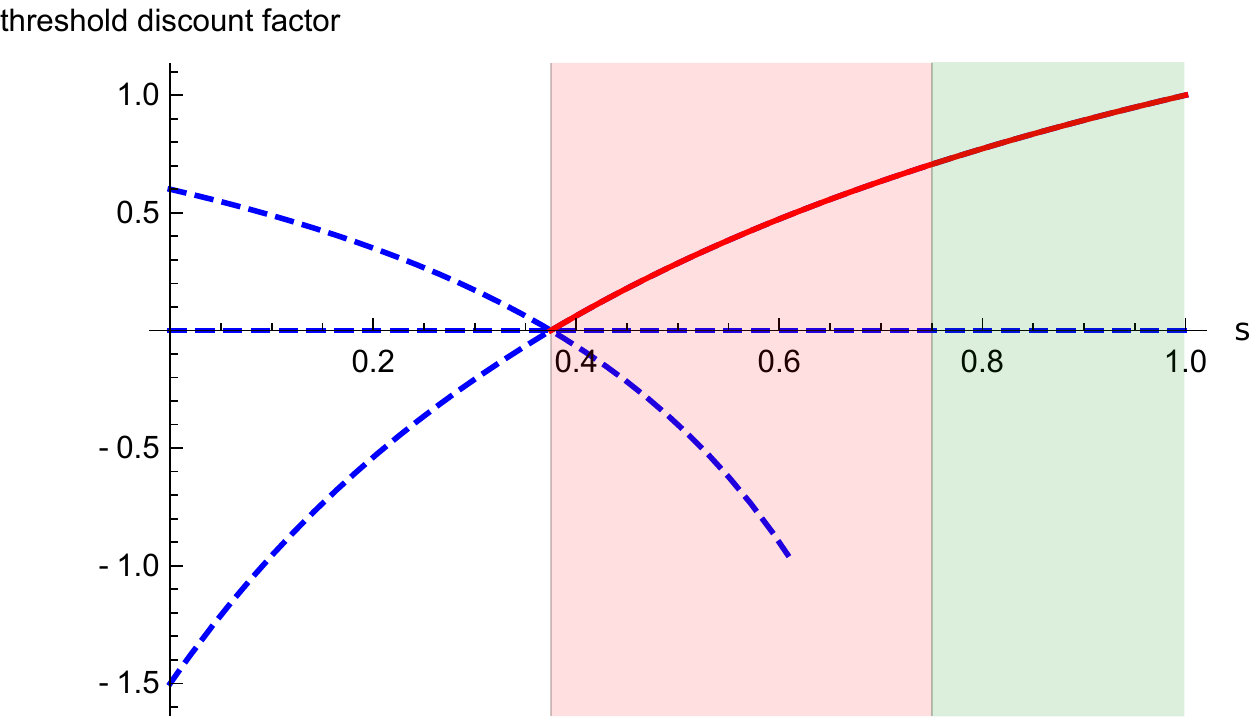}
\caption{\textcolor{black}{The red curve shows threshold discount factors for generous and extortionate strategies in the public goods game with $c=1,r=2, n=5$. Extortionate and generous strategies exist from $s=\frac{3}{8}$. In the region $\frac{3}{8}\leq s\leq\frac{3}{4}$, extortionate ZD strategies are a symmetric Nash equilibrium, while in the region $\frac{3}{4}\leq s<1$ generous ZD strategies are a symmetric Nash equilibrium. } }
 \label{fig: Thresholds generous and extortionate PGG}
\end{figure}
\begin{figure}
    \centering
 \includegraphics[width=0.8\linewidth]{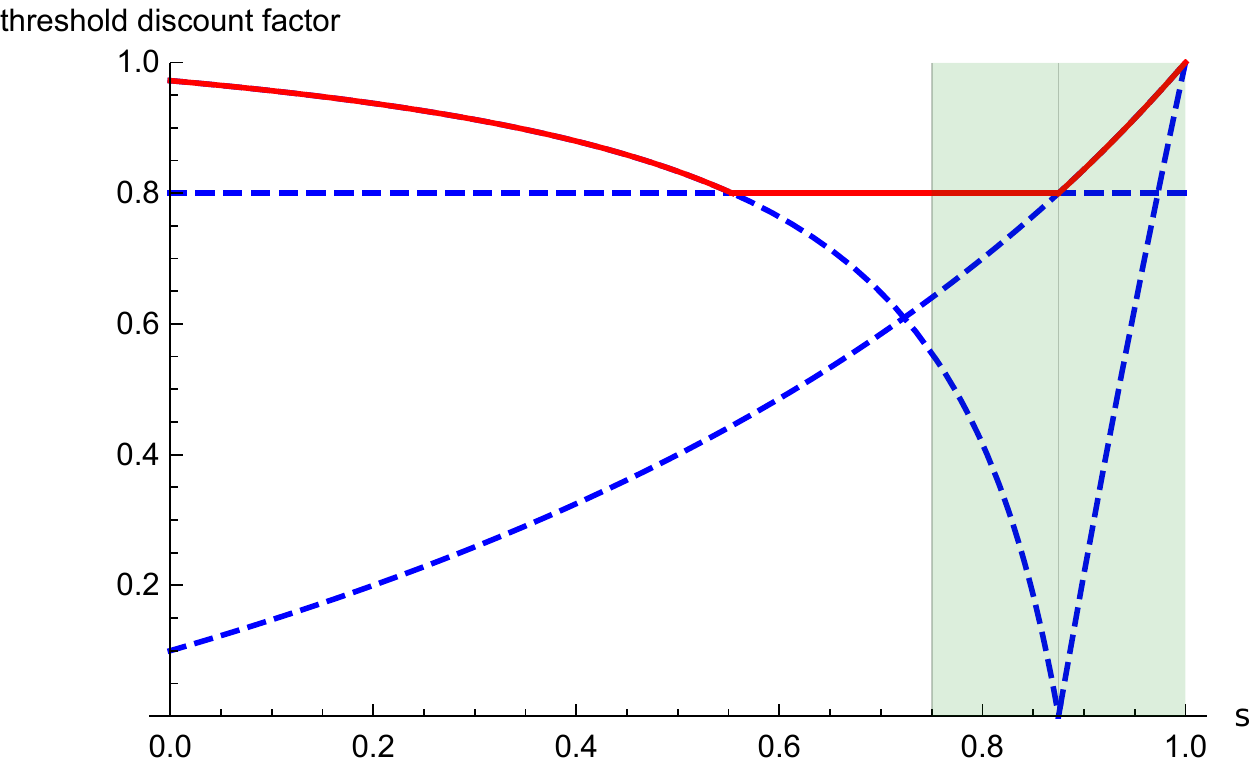}
    \caption{\textcolor{black}{The red curve shows threshold discount factors for generous strategies in the multiplayer snowdrift game with $b=2$, $c=1$ and $n=5$. Extortionate ZD strategies can enforce slopes from $s\geq \frac{7}{8}$ with the same discount factor as generous strategies, however only the latter are a Nash equilibrium, as indicated by the region $\frac{3}{4}\leq s<1 $.}}
    \label{fig:my_label}
\end{figure}

  \color{black}  
One can notice that when $s=1$, the threshold discount factor in \eqref{eq: threshold pgg} evaluates as $\delta_\tau=1$. 
This is consistent with Theorem \ref{Enforceable} and illustrates that fair strategies can only be enforced when the expected number of rounds is infinite (see Figure \ref{fig: Thresholds generous and extortionate PGG} for a numerical example). 
\textcolor{black}{From Propositions \ref{prop: ex. ext. pgg} and \ref{prop: thresholds pgg} one can also obtain insight in the effect of $\delta$ on the Nash equilibrium of the repeated public goods game. 
The result in \cite[SI Proposition 3]{hilbe2014cooperation} ensures that extortionate strategies are a symmetric Nash equilibrium if and only if $s<\frac{n-2}{n-1}$}, while generous strategies are a symmetric Nash equilibrium if and only if $s>\frac{n-2}{n-1}$. 
At $s=\frac{n-2}{n-1}$ both types of ZD strategies are a Nash equilibrium.
Combining this with the lower-bound of enforceable slopes and the expression for the threshold discount factor it follows that extortionate ZD strategies are a Nash equilibrium in the repeated public goods game from any discount factor $\delta>0$ provided that the slope is sufficiently small.
On the other hand, generous ZD strategies can only be a Nash equilibrium in the public goods game when the discount factor is sufficiently high $\delta\geq\frac{(n-r)(n-1)}{(n-1)^2+(r-1)}$. 
These conditions also hold when the deviating player is not restricted to ZD strategies and thus provide rather general equilibrium conditions for the repeated public goods game with discounting.

Let us now investigate the multiplayer snowdrift game in which the cost of cooperation is shared by all cooperators, resulting in a nonlinear payoff function with respect to the number of cooperating co-players $z$. 
The following proposition characterizes the enforceable payoff relations of generous and extortionate strategies by applying Theorem \ref{Enforceable} to the payoffs in Example \ref{example: NSD}.

\begin{proposition}[Enforceable slopes in the multiplayer snowdrift game]\label{prop: ex. ext. nsd}
	Suppose $p_0=0$, $l=0$ and $0<s<1$. 
	For the multiplayer snowdrift game with $b>c>0$, extortionate strategies can enforce any $s\geq 1-\frac{c}{b(n-1)}$. 
	Generous strategies, with $p_0=1$ and $l=b-\frac{c}{n}$, can enforce any $0<s<1$ independent of $n$.
\end{proposition}

The possibilities for extortion are thus limited by the payoff parameters $b$, $c$ and the group size $n$. 
In fact, the lower-bound on the enforceable slopes of extortionate strategies in Proposition \ref{prop: ex. ext. nsd} prevents extortionate strategies from being a Nash equilibrium in the multiplayer snowdrift game.
In contrast, generous strategies can enforce any slope \textit{independent} of the game parameters. 
However, the threshold discount factors of these strategies \textit{do} depend on the game parameters as is characterized in the following proposition.

\begin{proposition}[]\label{prop: thresholds generosity nsd}
	For the multiplayer snowdrift game with $b>c$ and $n\geq2$, for slopes $s\leq 1-\frac{c}{b(n-1)}$ the threshold discount factor for generous strategies is determined by
	\begin{equation}\label{eq: threshold nsd gen}
	\delta_\tau=  \textup{max}\left\{\frac{n-1}{n},\frac{(1-s)b-\frac{c}{n-1}}{(1-s)(b-\frac{c}{n})}\right\}.
	\end{equation}
	For higher slopes $s> 1-\frac{c}{b(n-1)}$ the threshold of generous strategies is determined by 	\begin{equation}\label{eq: thresholds ext nsd}
	\delta_\tau=   \frac{(1 - s)(\frac{c}{n}-c)+c}{(1 - s)(b - c)+c}. 
 	\end{equation}
 	The threshold discount factor of enforceable slopes of extortionate strategies are also given by \eqref{eq: thresholds ext nsd}.
\end{proposition}

Note that again the expression of the threshold discount factor for high slopes of both generous and extortionate strategies in \eqref{eq: thresholds ext nsd} becomes one when $s=1$. 
Because generous strategies have no restriction on the enforceable slopes, they can also be a Nash equilibrium provided that the slope and discount factor are large enough (see Figure \ref{fig:my_label} for a numerical example). 

In both the public goods game and the multiplayer snowdrift game the set of enforceable payoff relations, and thus the degree of extortion and generosity, are strongly influenced by the game parameters and the discount factor. 
Up to now, we focused on deriving the minimum expected number of rounds that are necessary to enforce some desired or given payoff relation. 
However, the examples in this section also illustrate the reverse problem: given an expected number of rounds or discount factor, what is the set of payoff relations that a single player can exert? And does the employed strategy constitute a Nash equilibrium? 
Here, we have answered these question for two well known multiplayer games but using our theoretical results many other games including the prisoner's dilemma, the volunteer's dilemma, and the multiplayer stag-hunt game, can be analysed in the same manner.

\color{black}
\section{Proofs of the main results}\label{sec: proofs of the main results}

\textcolor{black}{This section provides detailed proofs of the main results in Sections \ref{main results} and \ref{Threshold}.}

\subsection{Proof of Proposition \ref{necessary}}
Suppose all players are cooperating e.g. $\sigma=(C,C,\dots,C)$. Then from the definition of $\delta\mathbf{p}$ in equation \eqref{ZDrewritten} and the payoffs given in Table \ref{multiplayer payoff}, it follows that
\begin{equation}\label{allc}
    \delta\mathbf{p}_{(C,C,\dots,C)} = 1+\phi(1-s)(l-a_{n-1})-(1-\delta)p_0.
\end{equation}
Now suppose that all players are defecting. Similarly, we have
\begin{equation}\label{alld}
    \delta\mathbf{p}_{(D,D,\dots,D)} = \phi(1-s)(l-b_0)-(1-\delta)p_0.
\end{equation}
In order for these payoff relations to be enforceable, it needs to hold that both entries in equations \eqref{allc} and \eqref{alld} are in the interval $[0,\delta]$. Equivalently,
\begin{equation}\label{allcneq}
    (1-\delta)(1-p_0)\leq\phi(1-s)(a_{n-1}-l)\leq 1-(1-\delta)p_0,
\end{equation}
and
 \begin{equation}\label{alldneq}
     0\leq p_0(1-\delta)\leq\phi(1-s)(l-b_0)\leq\delta+(1-\delta)p_0
 \end{equation}
 Combining \eqref{allcneq} and \eqref{alldneq} it follows that
$0<(1-\delta)\leq\phi(1-s)(a_{n-1}-b_0).$
From the assumption that $a_{n-1}>b_0$ listed in Assumption \ref{assumptions}, it follows that 
 \begin{equation}\label{1-sstrict}
     0<\phi(1-s).
 \end{equation}
 
 Now suppose there is a single defecting player, i.e., $\sigma=(C,C,\dots,D)$ or any of its permutations. In this case, the entries of the memory-one strategy are as given in equation \eqref{sigmaentries}. Again, for both cases we require  $\delta\mathbf{p}_{\sigma}$ to be in the interval $[0,\delta]$. This results in the inequalities given in equations \eqref{CDi} and \eqref{CDj}.
\begin{figure*}
\begin{equation}\label{sigmaentries}
\begin{aligned}
    \delta\mathbf{p}_{\sigma}=\begin{cases}
    1+\phi[sa_{n-2}-(1-w_j)a_{n-2}-w_jb_{n-1}+(1-s)l]-(1-\delta)p_0, \ &\text{if defector is $j\neq i$};\\
    \phi[sb_{n-1}-a_{n-2}+(1-s)l]-(1-\delta)p_0, \ &\text{if defector is $i$.}
    \end{cases}
    \end{aligned}
\end{equation}
\begin{equation}\label{CDi}
    0\leq p_0(1-\delta)\leq\phi[sb_{n-1}-a_{n-2}+(1-s)l]\leq\delta+(1-\delta)p_0
\end{equation}
\begin{equation}\label{CDj}
(1-\delta)(1-p_0)\leq\phi[-sa_{n-2}+(1-w_j)a_{n-2}+w_jb_{n-1}-(1-s)l]\leq 1-p_0(1-\delta)
\end{equation}
\end{figure*}
By combining the equations \eqref{CDi} and \eqref{CDj} we obtain
\begin{equation}
   0< (1-\delta)\leq\phi(s+w_j)(b_{n-1}-a_{n-2}).
\end{equation}
Again, because of the assumption $b_{z+1}>a_z$ it follows that
\begin{equation}\label{swjstrict}
    0<\phi(s+w_j),\ \forall j\neq i.
\end{equation}
The inequalities \eqref{swjstrict} and \eqref{1-sstrict} together imply that
\begin{equation}\label{1+w_jstrict}
    0<\phi(1+w_j), \forall j\neq i.
\end{equation}
Because at least one $w_j>0$, it follows that 
\begin{equation}\label{phi>0}
    \phi>0.
\end{equation}
Combining with equation \eqref{1-sstrict} we obtain
\begin{equation}\label{s<1}
  s<1. 
\end{equation}
In combination with equation \eqref{1+w_jstrict} it follows that 
\begin{equation} \label{minwj}
    \begin{aligned}
        \forall j\neq i: s+w_j>0\Leftrightarrow \forall j\neq i: w_j>-s \Leftrightarrow \underset{j\neq i}{\text{min}}\ w_j>-s.
    \end{aligned}
\end{equation}
The inequalities in the equations \eqref{s<1} and \eqref{minwj} finally produce the bounds on $s$:
\begin{equation}
    -\underset{j\neq i}{\text{min}}\ w_j<s<1.
\end{equation}
Moreover, because it is required that $\sum_{j=1}^nw_j=1$, it follows that $\underset{j\neq i}{\mathrm{min}}\ w_j\leq\frac{1}{n-1}$. Hence the necessary condition turns into:
\begin{equation}\label{bounds s}
    -\frac{1}{n-1}\leq -\underset{j\neq i}{\text{min}} \ w_j<s< 1.
\end{equation}
We continue to show the necessary upper and lower bound on $l$. 
From equation \eqref{allcneq} we obtain:
\begin{equation}\label{l-an-1}
   \phi(1-s)(l-a_{n-1}) \leq (1-p_0)(\delta-1)\leq 0.
\end{equation}
From equation \eqref{1-sstrict} we know $\phi(1-s)>0$. Together with equation \eqref{l-an-1} this implies the necessary condition
\begin{equation} 
    l-a_{n-1}\leq 0 \Leftrightarrow l\leq a_{n-1}.
\end{equation}
We continue with investigating the lower-bound on $l$, from equation \eqref{alldneq} 
\begin{equation}
0\leq p_0(1-\delta)\leq\phi(1-s)(l-b_0)\leq \delta+(1-\delta)p_0.
\end{equation}
Because $\phi(1-s)>0$ it follows that 
$l\geq b_0.$
Naturally, when $l=a_{n-1}$ by assumption \ref{assumptions} it holds that $l>b_0$ and when $l=b_0$ then $l<a_{n-1}$. This completes the proof.~\hfill \qedsymbol

\subsection{Proof of Lemma \ref{p0con}}
For brevity, in the following proof we refer to equations that are found in the proof of Proposition \ref{necessary}.
Assume the ZD strategy is extortionate, hence $l=b_0$. From the lower bound in \eqref{alldneq} in order for $l$ to be enforceable, it is necessary that $p_0=0$. This proves the first statement.
 Now assume the ZD strategy is generous, hence $l=a_{n-1}$. From the lower bound in \eqref{allcneq} in order for $l$ to be enforceable, it is necessary that $p_0=1$. This proves the second statement and completes the proof.~\hfill\qedsymbol

\subsection{Proof of Theorem \ref{Enforceable}}\label{subsec: proof of theorem enforceable}

In the following we refer to the key player, who is employing the ZD strategy, as player $i$.
Let $\sigma=(x_1,\dots,x_n)$ such that $x_k\in\A$ and let $\sigma^C$ be the set of $i's$ co-players that cooperate and let $\sigma^D$ be the set of $i's$ co-players that defect. 
Also, let $|\sigma|$ be the total number of cooperators in $\sigma$ including player $i$.
Using this notation, for some action profile $\sigma$ we may write the ZD strategy as
\begin{equation}
    \delta \mathbf{p}_\sigma= \mathbf{p}^{\mathrm{rep}}+\phi[(1-s)(l-g^i_\sigma)+\sum_{j\neq i} ^nw_j(g_\sigma^i-g_\sigma^j)]-(1-\delta)p_0.
\end{equation}
Also, note that
\begin{equation}\label{sumwjgj}
    \sum_{j\neq i}^nw_jg_\sigma^j=\sum_{k\in\sigma^D}w_kg^k_{\sigma}+\sum_{h\in\sigma^C}w_hg_\sigma^h,
\end{equation}
and because $\sum_{j\neq i}^nw_j=1$ it holds that
$ \sum_{l\in\sigma^C}w_l=1-\sum_{k\in\sigma^D}w_k.$
Substituting this into equation \eqref{sumwjgj} and using the payoffs as in Table \ref{multiplayer payoff} we obtain
  $ \sum_{j\neq i}^nw_jg_\sigma^j= a_{|\sigma|-1}+\sum_{j\in\sigma^D}w_j(b_{|\sigma|}-a_{|\sigma|-1}).$
Accordingly, the entries of the ZD strategy $\delta\mathbf{p}_\sigma$ are given by equation \eqref{sigma p}. For all $\sigma\in\boldsymbol{\A}$ we require that
\begin{equation}\label{unit}
0\leq\delta\mathbf{p}_{\sigma}\leq \delta. 
\end{equation}
This leads to the inequalities in equations \eqref{si=c bounds} and \eqref{si=d bounds}.
Because $\phi>0$ can be chosen arbitrarily small, the inequalities in equation \eqref{si=c bounds} can be satisfied for some $\delta\in(0,1)$ and $p_0\in[0,1]$ if and only if for all $\sigma$ such that $x_i=C$ the inequalities in equation \eqref{upper l} are satisfied.
\begin{figure*}
\begin{equation}\label{sigma p}
\begin{aligned}
    \delta\mathbf{p}_{\sigma}=\begin{cases}
    1+\phi\left[ (1-s)(l-a_{|\sigma|-1})-\sum\limits_{j\in\sigma^D}w_j(b_{|\sigma|}-a_{|\sigma|-1})\right]-(1-\delta)p_0,\ &\text{if}\ x_i=C, \\
    \phi\left[ (1-s)(l-b_{|\sigma|})+\sum\limits_{j\in\sigma^C}w_j(b_{|\sigma|}-a_{|\sigma|-1})\right]-(1-\delta)p_0,\ &\text{if} \ x_i=D.
    \end{cases}
    \end{aligned}
\end{equation}
\begin{equation}\label{si=c bounds}
0\leq (1-\delta)(1-p_0)\leq\phi\left[ (1-s)(a_{|\sigma|-1}-l)+\sum\limits_{j\in\sigma^D}w_j(b_{|\sigma|}-a_{|\sigma|-1})\right]\leq1-(1-\delta)p_0
\end{equation}
\begin{equation}\label{si=d bounds}
    0\leq(1-\delta)p_0\leq\phi\left[ (1-s)(l-b_{|\sigma|})+\sum\limits_{j\in\sigma^C}w_j(b_{|\sigma|}-a_{|\sigma|-1})\right]\leq\delta+(1-\delta)p_0.
\end{equation}
\end{figure*}

\begin{equation}\label{upper l}
    0\leq (1-s)(a_{|\sigma|-1}-l)+\sum\limits_{j\in\sigma^D}w_j(b_{|\sigma|}-a_{|\sigma|-1}).
\end{equation}
The inequality \eqref{upper l} together with the necessary condition $s<1$ from Proposition \ref{necessary} implies that 
\begin{equation}\label{upper l 2}
    a_{|\sigma|-1}+\frac{\sum\limits_{j\in\sigma^D}w_j(b_{|\sigma|}-a_{|\sigma|-1})}{(1-s)}\geq l,
\end{equation}
and thus provides an upper-bound on the enforceable baseline payoff $l$.
We now turn our attention to the inequalities in equation \eqref{si=d bounds} that can be satisfied if and only if for all $\sigma$ such that $x_i=D$ the following holds
\begin{equation}\label{lower l}
\begin{aligned}
        0\leq (1-s)(l-b_{|\sigma|})+\sum\limits_{j\in\sigma^C}w_j(b_{|\sigma|}-a_{|\sigma|-1})\\
        \xRightarrow{(1-s)>0}
        b_{|\sigma|}-\frac{\sum\limits_{j\in\sigma^C}w_j(b_{|\sigma|}-a_{|\sigma|-1})}{(1-s)}\leq l.
        \end{aligned}
\end{equation}

Combining equations \eqref{lower l} and \eqref{upper l 2} we obtain

\begin{equation}\label{bounds l total}
\begin{aligned}
     \underset{|\sigma| s.t. x_i=D}{\text{max}}\left\{ b_{|\sigma|}-\frac{\sum\limits_{l\in\sigma^C}w_l(b_{|\sigma|}-a_{|\sigma|-1})}{(1-s)} \right\} \leq l,\\
     l\leq \underset{|\sigma| s.t. x_i=C}{\text{min}}\left\{ a_{|\sigma|-1}+\frac{\sum\limits_{k\in\sigma^D}w_k(b_{|\sigma|}-a_{|\sigma|-1})}{(1-s)} \right\}.
\end{aligned}
\end{equation}

Because $b_{|\sigma|}-a_{|\sigma|-1}>0$ and $(1-s)>0$  the minima and maxima of the bounds in equation \eqref{bounds l total} are achieved by choosing the $w_j$ as small as possible. That is, the extrema of the bounds on $l$ are achieved for those states $\sigma|_{x_i=D}$ in which $\sum\limits_{l\in\sigma^C}w_l$ is minimum and those $\sigma|_{x_i=C}$ in which $\sum\limits_{k\in\sigma^D}w_k$ is minimum. 
 Let $\hat{w}_z=\underset{w_h\in w}{\text{min}}(\sum_{h=1}^zw_h)$ denote the sum of the $j$ smallest weights and let $\hat{w}_0=0$. By the above reasoning, equation \eqref{bounds l total} can be equivalently written as in the theorem in the main text. 
 Now, suppose we have a \textit{non-strict} upper-bound on the base-level payoff, i.e.,
 $$l=a_{|\sigma|-1}+\frac{\sum\limits_{k\in\sigma^D}w_k(b_{|\sigma|}-a_{|\sigma|-1})}{(1-s)}.$$ 
From equation \eqref{si=c bounds} it follows that $p_0=1$ is required. Then equation \eqref{si=d bounds} implies
\begin{equation}
\begin{aligned}
     0<(1-s)(l-b_{|\sigma|})+\sum\limits_{j\in\sigma^C}w_j(b_{|\sigma|}-a_{|\sigma|-1})\\
     \xRightarrow{(1-s)>0}
     b_{|\sigma|}-\frac{\sum\limits_{j\in\sigma^C}w_j(b_{|\sigma|}-a_{|\sigma|-1})}{(1-s)}<l. 
\end{aligned}
\end{equation} 

Which is exactly the corresponding lower-bound of $l$, that is thus required to be strict when the upper-bound is non-strict. 

Now suppose we have a non-strict lower bound, e.g. $$l=b_{|\sigma|}-\frac{\sum\limits_{l\in\sigma^C}w_l(b_{|\sigma|}-a_{|\sigma|-1})}{(1-s)}.$$ 
From equation \eqref{si=d bounds} it follows that $p_0=0$ is required. Then, the inequalities in equation \eqref{si=c bounds} require that
\begin{equation}
\begin{aligned}
0<(1-s)(a_{|\sigma|-1}-l)+\sum\limits_{j\in\sigma^D}w_j(b_{|\sigma|}-a_{|\sigma|-1})\\
\xRightarrow{(1-s)>0}
a_{|\sigma|-1} +\frac{\sum\limits_{j\in\sigma^D}w_j(b_{|\sigma|}-a_{|\sigma|-1})}{(1-s)}>l.
\end{aligned}
\end{equation}
This completes the proof.~\hfill\qedsymbol

\subsection{Proof of Theorem \ref{prop: extortionate thresholds}}
For brevity in the following proof we refer to equations that can be found in the proof of Theorem \ref{Enforceable}. 
From Lemma \ref{p0con} we know that in order for the extortionate payoff relation  to be enforceable it is necessary that $p_0=0$. 
By substituting this into equation \eqref{si=c bounds} it follows that in order for the payoff relation to be enforceable it is required that for all $\sigma$ such that $x_i=C$ the following holds:
\begin{equation}\label{rhoC>0}
   \rho^C(\sigma)=(1-s)(a_{|\sigma|-1}-l)+\sum_{j\in\sigma^D}w_j(b_{|\sigma|-a_{|\sigma|-1}})>0.
\end{equation}
Hence, Equation \eqref{si=c bounds} with $p_0=0$ implies that for all $\sigma$ such that $x_i=C$ it holds that
\begin{equation}\label{underlinerhoc}
    \frac{1-\delta}{\rho^C(\sigma)}\leq \phi\leq \frac{1}{\rho^C(\sigma)}\Rightarrow \frac{1-\delta}{\underline{\rho}^C(z,\hat{w}_z)}\leq \phi\leq \frac{1}{\overline{\rho}^C(z,\tilde{w}_z)} .
\end{equation}
Naturally, $\overline{\rho}^C\geq\underline{\rho}^C$. In the special case in which equality holds, it follows from equation \eqref{underlinerhoc} that $\delta\geq0$, which is true by definition of $\delta$. 
We continue to investigate the case in which $\overline{\rho}^C>\underline{\rho}^C$.
In this case, a solution to equation \eqref{underlinerhoc} for some $\phi>0$ exists if and only if
\begin{equation}\label{barrhoc-rhoc}
    \frac{1-\delta}{\underline{\rho}^C(z,\hat{w}_z)}\leq \frac{1}{\overline{\rho}^C(z,\tilde{w}_z)}\Rightarrow  \delta\geq\frac{\overline{\rho}^C-\underline{\rho}^C}{\overline{\rho}^C},
\end{equation}
which leads to the first expression in the theorem.
Now, from equation \eqref{si=d bounds} with $p_0=0$, it follows that in order for the payoff relation to be enforceable it is necessary that
\begin{equation}\label{overlinerhod}
 \forall \sigma \text{\ s.t.\ } x_i=D:\quad 0\leq\phi\rho^D(\sigma)\leq \delta \Rightarrow 0\leq\phi\overline{\rho}^D(z,\tilde{w}_z)\leq \delta.
 \end{equation}
Because $\phi>0$ is necessary for the payoff relation to be enforceable, it follows that $\rho^D(\sigma)\geq 0$ for all $\sigma$ such that $x_i=D$. 
Let us first investigate the special case in which $\overline{\rho}^D(z,\tilde{w}_z)=0$. 
Then \eqref{overlinerhod} is satisfied for any $\phi>0$ and $\delta\in(0,1)$.
Now, assume $\overline{\rho}^D(z,\tilde{w}_z)>0$. Then, equations \eqref{overlinerhod} and \eqref{underlinerhoc} imply
\begin{equation}\label{eq: rho bounds ext}
\frac{1-\delta}{\underline{\rho}^C(z,\hat{w}_z)}\leq \phi\leq\frac{\delta}{\overline{\rho}^D(z,\tilde{w}_z)}.
\end{equation}
In order for such a $\phi$ to exist it needs to hold that 
 \begin{equation}\label{deltageqrhod}
      \frac{1-\delta}{\underline{\rho}^C(z,\hat{w}_z)}\leq \frac{\delta}{\overline{\rho}^D(z,\tilde{w}_z)}\xRightarrow{\overline{\rho}^D,\ \underline{\rho}^C>0} \delta\geq\frac{\overline{\rho}^D}{\overline{\rho}^D+\underline{\rho}^C}.
 \end{equation}
 This completes the proof.~\hfill\qedsymbol

\subsection{Proof of Theorem \ref{prop: generous thresholds}}
The proof is similar to the extortionate case in the proof of Theorem \ref{prop: extortionate thresholds}.
From Lemma \ref{p0con} we know that in order for the generous payoff relation to be enforceable it is necessary that $p_0=1$. By substituting this into equation \eqref{si=d bounds} it follows that in order for the payoff relation to be enforceable it is required that for all $\sigma$ such that $x_i=D$ the following holds:
\begin{equation}\label{rhoC>0d}
   \rho^D(\sigma)=(1-s)(l-b_{|\sigma|})+\sum_{j\in\sigma^C}w_j(b_{|\sigma|-a_{|\sigma|-1}})>0.
\end{equation}
Hence, equation \eqref{si=d bounds} with $p_0=1$ implies that for all $\sigma$ such that $x_i=D$ it holds that
\begin{equation}\label{underlinerhocd}
    \frac{1-\delta}{\rho^D(\sigma)}\leq \phi\leq \frac{1}{\rho^D(\sigma)}\Rightarrow \frac{1-\delta}{\underline{\rho}^D(z,\hat{w}_z)}\leq \phi\leq \frac{1}{\overline{\rho}^D(z,\tilde{w}_z)}.
\end{equation}
If $\overline{\rho}^D=\underline{\rho}^D>0$ this implies $\delta\geq0$. Otherwise equation \eqref{underlinerhocd} implies that
\begin{equation}\label{one generous bound}
\frac{1-\delta}{\underline{\rho}^D(z,\hat{w}_z)}\leq \frac{1}{\overline{\rho}^D(z,\tilde{w}_z)} \Rightarrow \delta\geq\frac{\overline{\rho}^D-\underline{\rho}^D}{\overline{\rho}^D},
\end{equation}
which leads to the first expression in the theorem. 
Moreover, from equation \eqref{si=c bounds} we know that the following must hold:
\begin{equation}\label{rhocleqdelta}
    \forall \sigma \text{\ s.t.\ } x_i=C:\quad 0\leq\phi\rho^C(\sigma)\leq \delta \Rightarrow 0\leq\phi\overline{\rho}^C(z,\tilde{w}_z)\leq\delta.
\end{equation}
Because $\phi>0$ it follows that $\rho^C(\sigma)\geq 0$ for all $\sigma$ such that $x_i=C$. 
Let us now consider the special case in which $\overline{\phi}\rho^C(z,\tilde{w}_z)=0$. 
Then, equation \eqref{rhocleqdelta} is satisfied for any $\phi>0$ and $\delta\in(0,1)$.
Now suppose $\overline{\rho}^C(z,\tilde{w}_z)>0$. Then, \eqref{rhocleqdelta} and \eqref{underlinerhocd} imply that in order for the generous strategy to be enforceable it is necessary that
\begin{equation}
  \frac{1-\delta}{\underline{\rho}^D(z,\hat{w}_z)}\leq \phi\leq   \frac{\delta}{\overline{\rho}^C(z,\tilde{w}_z)}.
\end{equation}
Such a $\phi$ exists if and only if
\begin{equation}
    \frac{1-\delta}{\underline{\rho}^D(z,\hat{w}_z)}\leq\frac{\delta}{\overline{\rho}^C(z,\tilde{w}_z)}\xRightarrow{\underline{\rho}^D,\ \overline{\rho}^C>0}\delta\geq\frac{\overline{\rho}^C}{\underline{\rho}^D+\overline{\rho}^C}.
\end{equation}
This completes the proof.~\hfill\qedsymbol

\subsection{Proof of Theorem \ref{thm: equalizer strategies thresholds}}

	For brevity, we refer to equations found in the proof of Theorem \ref{Enforceable}. From \eqref{si=c bounds} and \eqref{si=d bounds} it follows that in order for the payoff relation to be enforceable for $p_0\in(0,1)$ it must hold that for all $\sigma$ such that $x_i=C$, $\rho^C(\sigma)>0,$ and for all $\sigma$ such that $x_i=D$, $\rho^D(\sigma)>0.$ For the existence of equalizer strategies this must also hold for the special case in which $s=0$. Hence, we can rewrite \eqref{si=c bounds} and \eqref{si=d bounds} to obtain the following set of inequalities
	\begin{align}
	\frac{(1-\delta)(1-p_0)}{\underline{\rho}^C(z,\hat{w}_z)}\leq&\phi\leq\frac{1-(1-\delta)p_0}{\overline{\rho}^C(z,\tilde{w}_z) },\\
	\frac{(1-\delta)p_0}{\underline{\rho}^D(z,\hat{w}_z)}\leq \phi &\leq \frac{\delta+(1-\delta)p_0}{\overline{\rho}^D(z,\tilde{w}_z) }.
	\end{align}
	There exists such a $\phi>0$ if and only if the following inequalities are satisfied
	\begin{align}
	\frac{(1-\delta)p_0}{\underline{\rho}^D(z,\hat{w}_z)}\leq\frac{\delta+(1-\delta)p_0}{\overline{\rho}^D(z,\tilde{w}_z) }\label{p0 1},\\
	\frac{(1-\delta)p_0}{\underline{\rho}^D(z,\hat{w}_z)}\leq\frac{1-(1-\delta)p_0}{\overline{\rho}^C(z,\tilde{w}_z) }\label{p0 2},\\
	\frac{(1-\delta)(1-p_0)}{\underline{\rho}^C(z,\hat{w}_z)}\leq \frac{1-(1-\delta)p_0}{\overline{\rho}^C(z,\tilde{w}_z) } \label{p0 3},  \\
	\frac{(1-\delta)(1-p_0)}{\underline{\rho}^C(z,\hat{w}_z)}\leq\frac{\delta+(1-\delta)p_0}{\overline{\rho}^D(z,\tilde{w}_z) }\label{p0 4}.	
	\end{align}
	By collecting the terms in $p_0$ and $\delta$ for \eqref{p0 1}-\eqref{p0 4} the conditions can be derived as follows.
	The condition in \eqref{p0 1} can be satisfied if and only if 
	$$p_0(1-\delta)\left(\overline{\rho}^D(z,\tilde{w}_z)-\underline{\rho}^D(z,\hat{w}_z)\right)\leq \underline{\rho}^D(z,\hat{w}_z)\delta.$$	
	In the special case that $\overline{\rho}^D(z,\tilde{w}_z)-\underline{\rho}^D(z,\hat{w}_z)=0$, this is satisfied for every $p_0\in(0,1)$ and $\delta\in(0,1)$. On the other hand, if $\overline{\rho}^D(z,\tilde{w}_z)-\underline{\rho}^D(z,\hat{w}_z)>0$, then the inequality can be satisfied for every $p_0\in(0,1)$ if and only if \eqref{req1} holds.
Likewise, \eqref{p0 3} can be satisfied if and only if
$$-p_0(1-\delta)\left(\overline{\rho}^C-\underline{\rho}^C\right)\leq \underline{\rho}^C-(1-\delta)\overline{\rho}^C.$$
If 	$\overline{\rho}^C-\underline{\rho}^C=0$, this inequality is satisfied for every $p_0\in(0,1).$
On the other hand, if $\overline{\rho}^C-\underline{\rho}^C>0$, the inequality is satisfied if and only if the condition in \eqref{req3} holds.
\eqref{p0 2} holds if and only if the condition in \eqref{req4} holds.
Finally, \eqref{p0 4} holds if and only if the condition in \eqref{req2} holds.

\section{Final remarks}\label{fw}
\textcolor{black}{
We have extended the existing results for ZD strategies to multiplayer social dilemmas with discounting.
However, the fundamental relation between the memory-one strategy and the mean distribution is independent of the structure and symmetry of the game and thus the results in this paper can be extended by considering discounted multiplayer games that are not social dilemmas or have asymmetric single-round payoffs. 
Our theory supports the finding that due to the finite expected number of rounds the initial probability to cooperate of the key player remains important for the opportunities to exert control.
Based on the necessary initial probability to cooperate we derived expressions for the minimum discount factor above which a ZD strategy can enforce some desired generous or extortionate payoff relation. 
Because equalizer strategies do not impose any conditions on the initial probability to cooperate, 
we have derived a condition that ensures the desired equalizer strategy is enforceable for a variable initial probability to cooperate in the open unit interval. 
By combining the set of enforceable payoff relations and the threshold discount factors our results can also be used to investigate under which conditions on the expected number of rounds, generous and extortionate ZD strategies, that both can promote mutual cooperation in social dilemmas, constitute a symmetric Nash equilibrium in the multiplayer social dilemma. Future research can include individual or time-varying discounting functions, and the analysis of subgame perfection of the ZD strategy Nash equilibria.}


\bibliography{references}
\bibliographystyle{abbrv}

\appendices
\section{Proof of Proposition \ref{prop: ex. ext. pgg}}\label{app: ex. ext. pgg}
By applying Theorem \ref{Enforceable} to the public goods game it follows that the enforceable baseline payoffs are 
 \begin{align}\label{pggboundsLow}
     \textup{max}\left\{ 0,\frac{rc(n-1)}{n}-\frac{c}{1-s} \right\} &\leq l,\\
     \textup{min}\left\{ \frac{rc}{n}-c+\frac{c}{1-s},rc-c\right\}&\geq l\label{pggboundsUp},
\end{align}
with at least one strict inequality.
For extortionate strategies we set $l=0$ and $0<s<1$. 
The inequalities in equations \eqref{pggboundsLow} and \eqref{pggboundsUp} become
\begin{align}
     \text{max}\left\{ 0,\frac{rc(n-1)}{n}-\frac{c}{1-s} \right\} \leq 0\label{extortionate pgg LOW}\\
     \text{min}\left\{ \frac{rc}{n}-c+\frac{c}{1-s},rc-c\right\} \geq 0\label{extortionate pgg UP}
\end{align}
Solving for $s$ will yield the enforceable slopes in the extortionate ZD strategy.
Observe that a necessary condition for equation \eqref{extortionate pgg LOW} to hold is that the left hand side is equal to $0$ and in order for this to hold it is required that
\begin{align}
n(r(1-s)-1)\leq r(1-s).\label{eq: n r s ext}
\end{align}
The conditions  $-\frac{1}{n-1}<s<1$ in Theorem \ref{Enforceable} and the assumption that $r$ is positive implies that $r(1-s)$ in the right-hand side of equation \eqref{eq: n r s ext} is required to be strictly positive. 
It follows that if $r(1-s)-1\leq0$ or equivalently $s\geq\frac{r-1}{r}$ the inequalities in equation \eqref{eq: n r s ext} are always satisfied. 
Note that if $s\geq\frac{r-1}{r}$ is satisfied, the left-hand side of the inequality in equation \eqref{extortionate pgg UP} reads as $rc-c$. 
It follows that for every $r>1$, every $s\geq\frac{r-1}{r}$ can be enforced independent of $n$.
On the other hand, when $s<\frac{r-1}{r}$ in order for equation \eqref{eq: n r s ext} to be satisfied it must hold that
\begin{equation}\label{eq: effect of n}
n\leq\frac{r(1-s)}{r(1-s)-1}.
\end{equation}
Note that $s<\frac{r-1}{r}$ implies $r(1-s)-1\neq 0$ so the above inequality in well-defined.
If \eqref{eq: effect of n} does not hold and $s<\frac{r-1}{r}$ than $\frac{rc(n-1)}{n}-\frac{c}{1-s}> 0$,
thus the lower-bound in equation \eqref{extortionate pgg LOW} is not satisfied and consequently there cannot exist extortionate strategies.
We now investigate the inequality in equation \eqref{extortionate pgg UP}. 
We already know that when $s\geq\frac{r-1}{r}$  the upper-bound reads as $0<rc-c$ and is satisfied for all $r>1$. Similarly, when $s<\frac{r-1}{r}$ and \eqref{eq: effect of n} holds the upper-bound reads as $0<rc-c$. 
We now move on to generous strategies with $l=rc-c$ and $0<s<1$. 
The inequalities in equations \eqref{pggboundsLow} and \eqref{pggboundsUp} become
\begin{align}
     \text{max}\left\{ 0,\frac{rc(n-1)}{n}-\frac{c}{1-s} \right\} &\leq rc-c, \label{generous pgg LOW}\\
     \text{min}\left\{ \frac{rc}{n}-c+\frac{c}{1-s},rc-c\right\} &\geq rc-c.\label{generous pgg UP}
\end{align}
Clearly in order for generous strategies to exist it is necessary that the left hand side of equation \eqref{generous pgg UP} reads as $rc-c$. 
Therefore it is required that
\begin{equation*}
    \frac{rc}{n}-c+\frac{c}{1-s}\geq rc-c \Leftrightarrow n(r(1-s)-1)\leq (1-s)r.
\end{equation*}
Hence, this condition is equivalent to the condition in equation \eqref{eq: n r s ext} and thus this condition gives the same feasible region for the existence of extortionate strategies. 
Now suppose that, $s<\frac{r-1}{r}$  and $n\geq\frac{r(1-s)}{r(1-s)-1}$. Also in this case, only equality is possible i.e. $n=\frac{r(1-s)}{r(1-s)-1}$ because otherwise the upper-bound is not satisfied. 
Next to this, if $s<\frac{r-1}{r}$ and $n=\frac{r(1-s)}{r(1-s)-1}$ in order for the lower-bound to be satisfied it is required that
$
    rc-c=\frac{rc}{n}-c+\frac{c}{1-s}\geq rc-c\geq0,
$
which is satisfied with a strict lower-bound for all $r>1$.~\hfill\qedsymbol

\section{Proof of Proposition \ref{prop: thresholds pgg}}\label{app: thresholds pgg}
For the linear public goods game the parameters in equation \eqref{eq: rho functions} can be obtained from the extrema of the following functions
\begin{equation}\label{eq: rho functions pgg}
\begin{aligned}
    \rho^C(z)&=(1-s)\left(\frac{rc(z+1)}{n}-c-l\right)+\frac{n-z-1}{n-1}c, \\
    \rho^D(z)&=(1-s)\left(l-\frac{rcz}{n}\right)+\frac{z}{n-1}c
    \end{aligned}
\end{equation}
We focus first prove the case in which $l=0$ and $0<s<1$, and thus the strategy is extortionate. 
In this case the equations in \eqref{eq: rho functions pgg} become
\begin{align}
    \rho_e^C(z):=(1-s)\left(\frac{rc(z+1)}{n}-c\right)+\frac{n-z-1}{n-1}c \label{eq:APP rho ext C}\\
    \rho_e^D(z):=-(1-s)\left(\frac{rcz}{n}\right)+\frac{z}{n-1}c\label{eq:APP rho ext D}
\end{align}
We continue to obtain the maximizers and minimizers of equations \eqref{eq:APP rho ext C} and \eqref{eq:APP rho ext D}, that because of linearity in $z$ can only occur at the extreme points $z=0$ and $z=n-1$. When $n>r$ and $r>1$, as is the case when the linear public goods game is a social dilemma, we have the following simple conditions on the slope of the extortionate strategy. If $-\frac{1}{n-1}<s\leq 1-\frac{n}{r(n-1)}$ no extortionate or generous strategies can exist. 
Hence assume $s\geq 1-\frac{n}{r(n-1)}$. Then,
\begin{equation}\label{eq: APP large s maximizers extortionate pgg}
    \begin{aligned}
    &\overline{\rho}_e^C=\rho_e^C(0)=(1-s)(\frac{rc}{n}-c)+c,\\ &\underline{\rho}_e^C=\rho_e^C(n-1)=(1-s)(rc-c)>0,\\
     &\overline{\rho}_e^D=\rho_e^D(n-1)=-(1-s)(\frac{rc(n-1)}{n})+c,\\ &\underline{\rho}_e^D=\rho_e^D(0)=0.
    \end{aligned}
\end{equation}
The fractions in Proposition \ref{prop: extortionate thresholds} become
\begin{align}
\frac{\overline{\rho}_e^D}{\overline{\rho}_e^D+\underline{\rho}_e^C}=\frac{\overline{\rho}_e^C-\underline{\rho}_e^C}{\overline{\rho}_e^C}&=\frac{(1-s)(\frac{r}{n}-r)+1}{(1-s)(\frac{r}{n}-1)+1}.
\end{align}
We focus now on the case in which $l=rc-c$ and $0<s<1$, and hence the strategy is generous. If $l=rc-c$ the equations in \eqref{eq: rho functions pgg} become
\begin{equation}\label{eq: rho functions pgg generous}
\begin{aligned}
    \rho_g^C(z)&:=(1-s)(\frac{rc(z+1)}{n}-rc)+\frac{n-z-1}{n-1}c \\
    \rho_g^D(z)&:=(1-s)(rc-c-\frac{rcz}{n})+\frac{z}{n-1}c
    \end{aligned}
\end{equation}
The extreme points of these functions read as 
\begin{equation}\label{eq: APP large s maximizers generous pgg}
    \begin{aligned}
    &\overline{\rho}_g^C=\rho_g^C(0)=\overline{\rho}_e^D,
    &\underline{\rho}_g^C=\rho_g^C(n-1)=\underline{\rho}_e^D,\\
    &\overline{\rho}_g^D=\rho_g^D(n-1)=\overline{\rho}_e^C, 
    &\underline{\rho}_g^D=\rho_g^D(0)=\underline{\rho}_e^C.
    \end{aligned}
\end{equation}
It follows that the fractions in Theorem \ref{prop: generous thresholds} are equivalent to those in Theorem \ref{prop: extortionate thresholds}.~\hfill\qedsymbol
\color{black}
\section{Proof of Proposition \ref{prop: ex. ext. nsd}}\label{app: ex. ext. nsd}
The following lemma characterizes the enforceable baseline payoffs in the multiplayer snowdrift game.
\begin{lemma}\label{lem: bounds nsd}
	For the multiplayer snowdrift game the enforceable baseline payoffs $l$ are determined as
	\begin{equation}\label{eq: bounds l NSD}
	\textup{max}\left\{0,b-\frac{c}{(n-1)(1-s)}\right\}\leq l\leq b-\frac{c}{n},
	\end{equation}
	with at least one strict inequality.
\end{lemma}

\begin{proof}
	Suppose $z=0$, then the inequalities in Theorem \ref{Enforceable} on the baseline payoff become
	\begin{equation}\label{eq: bounds l nSD z=0}
	0\leq l\leq b-c+\frac{c}{1-s}.
	\end{equation}
	And if $1\leq z\leq n-1$, the bounds on the enforceable baseline payoffs read as
	\begin{align}
	l&\geq b-\frac{c}{(n-1)(1-s)},\label{eq: NSD bounds LOW} \\
	l&\leq\underset{1\leq z\leq n-1}{\text{min}}\left\{ b-\frac{c}{z+1}+\frac{n-z-1}{n-1}\frac{c}{(z+1)(1-s)} \right\}.\label{eq:NSD bounds UP}
	\end{align}
	We continue to investigate the minimum upper-bound of $l$. 
	The upper-bound in \eqref{eq:NSD bounds UP} can be written as
	\begin{equation}
	l\leq \underset{1\leq z\leq n-1}{\text{min}}b+\underbrace{\frac{\left((n-1)s+1\right)c}{(n-1)(z+1)(1-s)}}_{:=\xi(z)}-\frac{c}{(n-1)(1-s)}.
	\end{equation}
	From Theorem \ref{Enforceable}, in order for a ZD strategy to exist it is necessary that $s<1$ and because in a multiplayer game $n>1$, the denominator of the fraction $\xi(z)$ is positive for all $0\leq z\leq n-1$. 
	Thus, if the numerator of $\xi(z)$ is positive as well, then the minimum of the upper-bound occurs when $z$ is maximum.
	Now because $c>0$ we have,
	$$[(n-1)s+1]c>0\Leftrightarrow (n-1)s+1>0 \Leftrightarrow s>-\frac{1}{n-1}.$$
	It follows from the bounds of enforceable slopes $s$ in Theorem \ref{Enforceable}, that $\xi(z)$ is indeed positive.
	Hence, for $1\leq z\leq n-1$ and enforceable slope $-\frac{1}{n-1}<s<1$, the minimum of the upper-bound occurs when all co-players are cooperating, i.e., $z=n-1$.
	In combination with the upper-bound in \eqref{eq: bounds l nSD z=0} for the case $z=0$ we obtain
	$ l \leq \text{min}\{b-\frac{c}{n},b-c+\frac{c}{1-s}\}$.
	Note that
	$b-\frac{c}{n}<b-c+\frac{c}{(1-s)} \Leftrightarrow  s>\frac{1}{1-n} \Leftrightarrow s>-\frac{1}{n-1}.$ Hence, for all enforceable slopes $-\frac{1}{n-1}<s<1$ we obtain
	$l\leq b-\frac{c}{n}.$
\end{proof}

The set of enforceable slopes for extortionate (resp. generous) strategies can be found by substituting $l=0$ (resp. $l=b-\frac{c}{n}$) into the \eqref{eq: bounds l NSD} and solve for $s$.\hfill\qedsymbol

\section{Proof of Proposition \ref{prop: thresholds generosity nsd}}\label{app: thr. gen. nsd}

The values of $\overline{\rho}^C(z)$ and $\underline{\rho}^C(z)$ from \eqref{eq: rho functions} for the multiplayer snowdrift game with equal weights are obtained from the extreme points of the following expression, for $0\leq z\leq n-1$:
\begin{equation}\label{eq: rho c functions nsd}
\begin{aligned}
\rho^C(z)&=(1-s)\left(b-\frac{c}{z+1}-l\right)+\frac{n-z-1}{n-1}\frac{c}{z+1}. \\
\end{aligned}
\end{equation}
For all enforceable slopes $-\frac{1}{n-1}<s<1$ the extreme points read as
\begin{equation}\label{eq: extrema rhoc}
\begin{aligned}
\overline{\rho}^C(z)&=\rho^C(0)=(1-s)(b-c-l)+c,\\
\underline{\rho}^C(z)&=\rho^C(n-1)=(1-s)(b-\frac{c}{n}-l).
\end{aligned}
\end{equation}
The values of $\overline{\rho}^D(z)$ and $\underline{\rho}^D(z)$ are obtained from the extreme points of the function
\begin{equation}\label{eq: z=0 rho functions nsd}
\rho^D(z)=\begin{cases}
(1-s)l, &\text{\ if\ } z=0\\
(1-s)(l-b)+\frac{c}{n-1}, &\text{\ if\ } z=1\dots n-1.
\end{cases}
\end{equation}
Assume $l=b-\frac{c}{n}$ and $0<s<1$ such that the ZD strategy is generous, from \eqref{eq: z=0 rho functions nsd} we obtain
	\begin{equation}
	\rho_g^D(z):=\begin{cases}
	(1-s)(b-\frac{c}{n}), &\text{\ if\ } z=0\\
	\frac{c}{n-1}-(1-s)\frac{c}{n}, &\text{\ if\ } z=1\dots n-1.
	\end{cases}
	\end{equation}
	For $s\leq 1-\frac{c}{b(n-1)}$ we have
	\begin{equation}\label{eq: app reg 1 small s maximizers generous nsd}
	\begin{aligned}
	&\overline{\rho}_g^D=\rho_g^D(0)=(1-s)(b-\frac{c}{n}),\
	&\underline{\rho}_g^D=\frac{c}{n-1}-(1-s)\frac{c}{n}.
	\end{aligned}
	\end{equation}
	Note that  $\underline{\rho}_g^D>0$ for all $s>-\frac{1}{n-1}$.
	Using the expressions in \eqref{eq: extrema rhoc} and substituting $l=0$ we also have
	
	\begin{equation}
	\begin{aligned}
	\overline{\rho}_g^C&=(1-s)\left(\frac{c}{n}-c\right)+c>0,\
	\underline{\rho}_g^C&=0.
	\end{aligned}
	\end{equation}
Again, 	note that $\overline{\rho}_g^C>0$ for all $s>-\frac{n}{n-1}$.
	The fractions in Theorem \ref{prop: generous thresholds} become
	\begin{align}
	\frac{\overline{\rho}_g^D-\underline{\rho}_g^D}{\overline{\rho}_g^D}&=\frac{(1-s)b-\frac{c}{n-1}}{(1-s)(b-\frac{c}{n})}, \ 
	\frac{\overline{\rho}_g^C}{\overline{\rho}_g^C+\underline{\rho}_g^D}&=\frac{n-1}{n}\label{eq: thres gen n case}
	\end{align}
	 We now continue to the case in which $1-\frac{c}{b(n-1)}<s<1$. 
	 Then, the extreme points are
	\begin{equation*}\label{eq: app reg 2 small s maximizers generous nsd}
	\begin{aligned}
	\overline{\rho}_g^D=\frac{c}{n-1}-(1-s)\frac{c}{n},\  \underline{\rho}_g^D=\rho_g^D(0)=(1-s)(b-\frac{c}{n}),
	\end{aligned}
	\end{equation*}
	and the fractions in Theorem \ref{prop: generous thresholds} become
	\begin{align}
	\frac{\overline{\rho}_g^D-\underline{\rho}_g^D}{\overline{\rho}_g^D}&=\frac{\frac{c}{n-1}-b(1-s)}{\frac{c}{n-1}-\frac{c}{n}(1-s)},\label{eq1}\\
	\frac{\overline{\rho}_g^C}{\overline{\rho}_g^C+\underline{\rho}_g^D}&=\frac{(1-s)(\frac{c}{n}-c)+c}{(1-s)(b-c)+c}.\label{eq2}
	\end{align}
	For all $0<s<1$ and $b>c>0$ the threshold in the Propositions follows.
 	The thresholds for extortionate strategies are obtained by substituting $l=0$ into \eqref{eq: extrema rhoc} and \eqref{eq: z=0 rho functions nsd} and find their extreme points for the enforceable slopes $s\geq1-\frac{c}{b(n-1)}$. The extreme points are:
	$$\overline{\rho}_e^D=\frac{c}{n-1}-(1-s)b\geq 0,\ 
	\underline{\rho}_e^D=0,$$ 
	$$\overline{\rho}_e^C=(1-s)(b-c)+c>0,
	\underline{\rho}_e^C=(1-s)(b-\frac{c}{n})>0.$$
	The threshold in the proposition is obtained by substituting these extreme points into the fractions of Theorem \ref{prop: extortionate thresholds}.

\end{document}